\documentclass[11pt,fullpage]{article}
\usepackage{theapa, rawfonts}

\usepackage[total={6.25in, 8.25in}]{geometry}
\usepackage{url}
\usepackage{tikz}
\usepackage{amsfonts}
\usepackage{graphicx,color}
\usepackage{xspace}
\usepackage{enumerate}
\usepackage{amssymb,amsmath}
\usepackage{amsthm}
\usepackage{caption}
\usepackage{subcaption}
\usepackage{marvosym}

\newtheorem{theorem}{Theorem}
\newtheorem{example}{Example}

\newtheorem{corollary}[theorem]{Corollary}

\newtheorem{lemma}[theorem]{Lemma}

\newtheorem{remark}[theorem]{Remark}

\usepackage{algorithmicx}
\usepackage{algorithm}

\usepackage{algpascal}
\usepackage{algc}
\usepackage[compatible]{algpseudocode}
\newcommand{\multilineee}[1]{%
	\begin{tabularx}{\dimexpr\linewidth-\ALG@thistlm}[t]{@{}X@{}}
		#1
	\end{tabularx}
}
\newcommand{\ngr}{$\mathit{Group}$\xspace}
\newcommand{\eb}{$\mathit{Empty}$\xspace}
\newcommand{\neb}{$\mathit{NonEmpty}$\xspace}

\algnewcommand\algorithmicforeach{\textbf{for each}}
\algdef{S}[FOR]{ForEach}[1]{\algorithmicforeach\ #1\ \algorithmicdo}

\algnewcommand{\Commentak}[1]{\Statex{\vspace{1mm}}\Comment{#1}}
\definecolor{cmntscolor}{rgb}{0.6, 0.6, 0.6}

\newcommand{\figTwo}{

	\tikzset{every picture/.style={line width=0.75pt}} 
	
	\begin{tikzpicture}[x=0.75pt,y=0.75pt,yscale=-1,xscale=1]
		
		\draw   (195.46,116.09) -- (227.89,116.09) -- (227.89,87.76) -- (195.46,87.76) -- cycle ;
		\draw   (118.46,172.07) -- (150.89,172.07) -- (150.89,86.82) -- (118.46,86.82) -- cycle ;
		\draw   (195.46,171.09) -- (227.89,171.09) -- (227.89,115.74) -- (195.46,115.74) -- cycle ;
		\draw   (41.46,172.09) -- (73.89,172.09) -- (73.89,41.27) -- (41.46,41.27) -- cycle ;
		\draw   (157.47,172.09) -- (189.89,172.09) -- (189.89,114.28) -- (157.47,114.28) -- cycle ;
		\draw   (80.46,75.09) -- (112.89,75.09) -- (112.89,27.59) -- (80.46,27.59) -- cycle ;
		\draw   (157.47,114.28) -- (189.89,114.28) -- (189.89,68.4) -- (157.47,68.4) -- cycle ;
		\draw   (157.47,114.28) -- (189.89,114.28) -- (189.89,172.07) -- (157.47,172.07) -- cycle ;
		\draw   (157.47,68.4) -- (189.89,68.4) -- (189.89,114.28) -- (157.47,114.28) -- cycle ;
		\draw  [fill={rgb, 255:red, 74; green, 144; blue, 226 }  ,fill opacity=0.2 ] (441.47,143.5) -- (473.89,143.5) -- (473.89,172.39) -- (441.47,172.39) -- cycle ;
		\draw   (118.46,86.84) -- (150.89,86.84) -- (150.89,172.07) -- (118.46,172.07) -- cycle ;
		\draw   (233.46,99.88) -- (265.89,99.88) -- (265.89,128.77) -- (233.46,128.77) -- cycle ;
		\draw   (233.46,71.22) -- (265.89,71.22) -- (265.89,99.88) -- (233.46,99.88) -- cycle ;
		\draw   (80.46,74.76) -- (112.89,74.76) -- (112.89,172.07) -- (80.46,172.07) -- cycle ;
		\draw   (233.46,128.77) -- (265.89,128.77) -- (265.89,171.07) -- (233.46,171.07) -- cycle ;
		\draw   (233.46,71.22) -- (265.89,71.22) -- (265.89,44.51) -- (233.46,44.51) -- cycle ;
		\draw   (157.47,27.62) -- (189.89,27.62) -- (189.89,172.09) -- (157.47,172.09) -- cycle ;
		\draw   (118.46,27.59) -- (150.89,27.59) -- (150.89,172.07) -- (118.46,172.07) -- cycle ;
		\draw   (233.46,26.59) -- (265.89,26.59) -- (265.89,171.07) -- (233.46,171.07) -- cycle ;
		\draw   (80.46,27.59) -- (112.89,27.59) -- (112.89,172.07) -- (80.46,172.07) -- cycle ;
		\draw  [fill={rgb, 255:red, 74; green, 144; blue, 226 }  ,fill opacity=0.2 ] (118.46,86.84) -- (150.89,86.84) -- (150.89,60.82) -- (118.46,60.82) -- cycle ;
		\draw   (195.46,115.74) -- (227.89,115.74) -- (227.89,171.07) -- (195.46,171.07) -- cycle ;
		\draw   (195.46,26.59) -- (227.89,26.59) -- (227.89,171.07) -- (195.46,171.07) -- cycle ;
		\draw  [fill={rgb, 255:red, 74; green, 144; blue, 226 }  ,fill opacity=0.2 ] (157.47,40.42) -- (189.89,40.42) -- (189.89,68.4) -- (157.47,68.4) -- cycle ;
		\draw   (195.46,59.77) -- (227.89,59.77) -- (227.89,87.76) -- (195.46,87.76) -- cycle ;
		\draw   (272.46,109.2) -- (304.89,109.2) -- (304.89,138.09) -- (272.46,138.09) -- cycle ;
		\draw   (272.46,80.54) -- (304.89,80.54) -- (304.89,109.2) -- (272.46,109.2) -- cycle ;
		\draw   (272.46,138.09) -- (304.89,138.09) -- (304.89,171.79) -- (272.46,171.79) -- cycle ;
		\draw   (272.46,27.32) -- (304.89,27.32) -- (304.89,171.79) -- (272.46,171.79) -- cycle ;
		\draw  [fill={rgb, 255:red, 74; green, 144; blue, 226 }  ,fill opacity=0.2 ] (118.46,45.51) -- (150.89,45.51) -- (150.89,60.82) -- (118.46,60.82) -- cycle ;
		\draw   (41.46,27.59) -- (73.89,27.59) -- (73.89,172.07) -- (41.46,172.07) -- cycle ;
		\draw   (41.46,27.59) -- (73.89,27.59) -- (73.89,41.27) -- (41.46,41.27) -- cycle ;
		\draw  [fill={rgb, 255:red, 74; green, 144; blue, 226 }  ,fill opacity=0.2 ] (80.46,74.76) -- (112.89,74.76) -- (112.89,27.27) -- (80.46,27.27) -- cycle ;
		\draw  [fill={rgb, 255:red, 74; green, 144; blue, 226 }  ,fill opacity=0.2 ] (157.47,172.09) -- (189.89,172.09) -- (189.89,114.28) -- (157.47,114.28) -- cycle ;
		\draw  [fill={rgb, 255:red, 74; green, 144; blue, 226 }  ,fill opacity=0.2 ] (41.46,172.07) -- (73.89,172.07) -- (73.89,41.25) -- (41.46,41.25) -- cycle ;
		\draw  [fill={rgb, 255:red, 74; green, 144; blue, 226 }  ,fill opacity=0.2 ] (441.47,143.5) -- (473.89,143.5) -- (473.89,85.68) -- (441.47,85.68) -- cycle ;
		\draw  [fill={rgb, 255:red, 74; green, 144; blue, 226 }  ,fill opacity=0.2 ] (118.46,172.09) -- (150.89,172.09) -- (150.89,86.84) -- (118.46,86.84) -- cycle ;
		\draw  [fill={rgb, 255:red, 74; green, 144; blue, 226 }  ,fill opacity=0.2 ] (523.46,173.07) -- (555.89,173.07) -- (555.89,87.82) -- (523.46,87.82) -- cycle ;
		\draw  [fill={rgb, 255:red, 74; green, 144; blue, 226 }  ,fill opacity=0.2 ] (481.46,173.09) -- (513.89,173.09) -- (513.89,42.27) -- (481.46,42.27) -- cycle ;
		\draw   (481.46,28.59) -- (513.89,28.59) -- (513.89,173.07) -- (481.46,173.07) -- cycle ;
		\draw   (481.46,28.59) -- (513.89,28.59) -- (513.89,42.27) -- (481.46,42.27) -- cycle ;
		\draw   (441.47,85.68) -- (473.89,85.68) -- (473.89,27.87) -- (441.47,27.87) -- cycle ;
		\draw   (523.47,87.68) -- (555.89,87.68) -- (555.89,29.87) -- (523.47,29.87) -- cycle ;
		\draw   (355.45,28.59) -- (387.87,28.59) -- (387.87,173.07) -- (355.45,173.07) -- cycle ;
		\draw  [fill={rgb, 255:red, 74; green, 144; blue, 226 }  ,fill opacity=0.2 ] (312.87,171.58) -- (345.3,171.58) -- (345.3,28.32) -- (312.87,28.32) -- cycle ;
		\draw  [fill={rgb, 255:red, 74; green, 144; blue, 226 }  ,fill opacity=0.2 ] (355.45,173.09) -- (387.88,173.09) -- (387.88,94.83) -- (355.45,94.83) -- cycle ;
		
		\draw (173.68,143.17) node   [align=left] {$ $4 $ $};
		\draw (173.68,91.34) node   [align=left] {$ $3$ $};
		\draw (133.55,129.61) node   [align=left] {$ $6};
		\draw (249.67,85.55) node   [align=left] {2};
		\draw (249.67,114.33) node   [align=left] {$ $2$ $};
		\draw (96.67,123.41) node   [align=left] {$ $7$ $};
		\draw (249.67,149.92) node   [align=left] {$ $3$ $};
		\draw (94.67,49) node   [align=left] {$ $3$ $};
		\draw (173.68,53.94) node   [align=left] {$ $2$ $};
		\draw (134.68,73.83) node   [align=left] {$ $2$ $};
		\draw (249.67,57.86) node   [align=left] {$ $2$ $};
		\draw (134.68,53.16) node   [align=left] {$ $1$ $};
		\draw (211.67,143.4) node   [align=left] {4};
		\draw (211.67,42.91) node   [align=left] {$ $2$ $};
		\draw (211.67,101.75) node   [align=left] {2};
		\draw (211.67,73.76) node   [align=left] {$ $2$ $};
		\draw (288.67,123.65) node   [align=left] {$ $2$ $};
		\draw (288.67,154.94) node   [align=left] {\begin{minipage}[lt]{15.31pt}\setlength\topsep{0pt}
				\begin{flushright}
					$ $2$ $
				\end{flushright}
				
		\end{minipage}};
		\draw (288.67,94.87) node   [align=left] {$ $2$ $};
		\draw (57.67,34.43) node   [align=left] {$ $1$ $};
		\draw (57.67,99.83) node   [align=left] {9};
		\draw (134.68,36.16) node   [align=left] {$ $1$ $};
		\draw (114,210) node [anchor=north west][inner sep=0.75pt]  [color={rgb, 255:red, 155; green, 155; blue, 155 }  ,opacity=1 ] [align=left] {\begin{minipage}[lt]{75.85pt}\setlength\topsep{0pt}
				\begin{flushright}
					\textsc{ProfilePacking}
				\end{flushright}
				
		\end{minipage}};
		\draw (58.17,183) node   [align=left] {$\displaystyle B_{1}$};
		\draw (94.52,183) node   [align=left] {$\displaystyle B_{2}$};
		\draw (134.52,183) node   [align=left] {$\displaystyle B_{3}$};
		\draw (172.52,183) node   [align=left] {$\displaystyle B_{4}$};
		\draw (251.52,181.73) node   [align=left] {$\displaystyle B_{6}$};
		\draw (212.52,181.73) node   [align=left] {$\displaystyle B_{5}$};
		\draw (290.52,182.73) node   [align=left] {$\displaystyle B_{7}$};
		\draw (457.68,114.59) node   [align=left] {$ $4 $ $};
		\draw (538.55,130.61) node   [align=left] {$ $6};
		\draw (457.68,157.94) node   [align=left] {\begin{minipage}[lt]{15.31pt}\setlength\topsep{0pt}
				\begin{flushright}
					$ $2$ $
				\end{flushright}
				
		\end{minipage}};
		\draw (497.67,100.83) node   [align=left] {9};
		\draw (478,211) node [anchor=north west][inner sep=0.75pt]  [color={rgb, 255:red, 155; green, 155; blue, 155 }  ,opacity=1 ] [align=left] {\begin{minipage}[lt]{35.71pt}\setlength\topsep{0pt}
				\begin{flushright}
					\textsc{First~Fit}
				\end{flushright}
				
		\end{minipage}};
		\draw (326.66,100.83) node   [align=left] {10};
		\draw (369.66,132.83) node   [align=left] {5};

	\end{tikzpicture}
}
\newcommand{\figOne}{

	\tikzset{every picture/.style={line width=0.75pt}} 
	
	\begin{tikzpicture}[x=0.75pt,y=0.75pt,yscale=-1,xscale=1]
		
		\draw   (195.46,116.09) -- (227.89,116.09) -- (227.89,87.76) -- (195.46,87.76) -- cycle ;
		\draw   (118.46,172.07) -- (150.89,172.07) -- (150.89,86.82) -- (118.46,86.82) -- cycle ;
		\draw   (195.46,171.09) -- (227.89,171.09) -- (227.89,115.74) -- (195.46,115.74) -- cycle ;
		\draw   (41.46,172.09) -- (73.89,172.09) -- (73.89,41.27) -- (41.46,41.27) -- cycle ;
		\draw   (157.47,172.09) -- (189.89,172.09) -- (189.89,114.28) -- (157.47,114.28) -- cycle ;
		\draw   (80.46,75.09) -- (112.89,75.09) -- (112.89,27.59) -- (80.46,27.59) -- cycle ;
		\draw   (157.47,114.28) -- (189.89,114.28) -- (189.89,68.4) -- (157.47,68.4) -- cycle ;
		\draw   (157.47,114.28) -- (189.89,114.28) -- (189.89,172.07) -- (157.47,172.07) -- cycle ;
		\draw   (157.47,68.4) -- (189.89,68.4) -- (189.89,114.28) -- (157.47,114.28) -- cycle ;
		\draw  [fill={rgb, 255:red, 74; green, 144; blue, 226 }  ,fill opacity=0.2 ] (157.47,39.5) -- (189.89,39.5) -- (189.89,68.39) -- (157.47,68.39) -- cycle ;
		\draw   (118.46,86.84) -- (150.89,86.84) -- (150.89,172.07) -- (118.46,172.07) -- cycle ;
		\draw   (233.46,99.88) -- (265.89,99.88) -- (265.89,128.77) -- (233.46,128.77) -- cycle ;
		\draw   (233.46,71.22) -- (265.89,71.22) -- (265.89,99.88) -- (233.46,99.88) -- cycle ;
		\draw   (80.46,74.76) -- (112.89,74.76) -- (112.89,172.07) -- (80.46,172.07) -- cycle ;
		\draw   (233.46,128.77) -- (265.89,128.77) -- (265.89,171.07) -- (233.46,171.07) -- cycle ;
		\draw   (233.46,71.22) -- (265.89,71.22) -- (265.89,44.51) -- (233.46,44.51) -- cycle ;
		\draw   (157.47,27.62) -- (189.89,27.62) -- (189.89,172.09) -- (157.47,172.09) -- cycle ;
		\draw   (118.46,27.59) -- (150.89,27.59) -- (150.89,172.07) -- (118.46,172.07) -- cycle ;
		\draw   (233.46,26.59) -- (265.89,26.59) -- (265.89,171.07) -- (233.46,171.07) -- cycle ;
		\draw   (80.46,27.59) -- (112.89,27.59) -- (112.89,172.07) -- (80.46,172.07) -- cycle ;
		\draw  [fill={rgb, 255:red, 74; green, 144; blue, 226 }  ,fill opacity=0.2 ] (118.46,86.84) -- (150.89,86.84) -- (150.89,60.82) -- (118.46,60.82) -- cycle ;
		\draw   (195.46,115.74) -- (227.89,115.74) -- (227.89,171.07) -- (195.46,171.07) -- cycle ;
		\draw   (195.46,26.59) -- (227.89,26.59) -- (227.89,171.07) -- (195.46,171.07) -- cycle ;
		\draw  [fill={rgb, 255:red, 74; green, 144; blue, 226 }  ,fill opacity=0.2 ] (195.46,87.76) -- (227.89,87.76) -- (227.89,115.74) -- (195.46,115.74) -- cycle ;
		\draw   (195.46,59.77) -- (227.89,59.77) -- (227.89,87.76) -- (195.46,87.76) -- cycle ;
		\draw   (272.46,109.2) -- (304.89,109.2) -- (304.89,138.09) -- (272.46,138.09) -- cycle ;
		\draw   (272.46,80.54) -- (304.89,80.54) -- (304.89,109.2) -- (272.46,109.2) -- cycle ;
		\draw   (272.46,138.09) -- (304.89,138.09) -- (304.89,171.79) -- (272.46,171.79) -- cycle ;
		\draw   (272.46,27.32) -- (304.89,27.32) -- (304.89,171.79) -- (272.46,171.79) -- cycle ;
		\draw   (118.46,45.51) -- (150.89,45.51) -- (150.89,60.82) -- (118.46,60.82) -- cycle ;
		\draw   (41.46,27.59) -- (73.89,27.59) -- (73.89,172.07) -- (41.46,172.07) -- cycle ;
		\draw   (41.46,27.59) -- (73.89,27.59) -- (73.89,41.27) -- (41.46,41.27) -- cycle ;
		\draw  [color={rgb, 255:red, 200; green, 200; blue, 200 }  ,draw opacity=1 ] (42,196) .. controls (42.03,200.67) and (44.37,202.99) .. (49.04,202.96) -- (163.06,202.32) .. controls (169.73,202.29) and (173.07,204.6) .. (173.1,209.27) .. controls (173.07,204.6) and (176.39,202.25) .. (183.06,202.21)(180.06,202.23) -- (297.08,201.58) .. controls (301.75,201.55) and (304.07,199.21) .. (304.05,194.54) ;
		\draw  [fill={rgb, 255:red, 74; green, 144; blue, 226 }  ,fill opacity=0.2 ] (80.46,74.76) -- (112.89,74.76) -- (112.89,27.27) -- (80.46,27.27) -- cycle ;
		\draw  [fill={rgb, 255:red, 74; green, 144; blue, 226 }  ,fill opacity=0.2 ] (118.46,45.51) -- (150.89,45.51) -- (150.89,60.82) -- (118.46,60.82) -- cycle ;
		\draw  [fill={rgb, 255:red, 74; green, 144; blue, 226 }  ,fill opacity=0.2 ] (157.47,172.09) -- (189.89,172.09) -- (189.89,114.28) -- (157.47,114.28) -- cycle ;
		\draw  [fill={rgb, 255:red, 74; green, 144; blue, 226 }  ,fill opacity=0.2 ] (41.46,172.07) -- (73.89,172.07) -- (73.89,41.25) -- (41.46,41.25) -- cycle ;
		\draw  [fill={rgb, 255:red, 74; green, 144; blue, 226 }  ,fill opacity=0.2 ] (195.46,173.9) -- (227.89,173.9) -- (227.89,116.09) -- (195.46,116.09) -- cycle ;
		\draw  [fill={rgb, 255:red, 74; green, 144; blue, 226 }  ,fill opacity=0.2 ] (118.46,172.09) -- (150.89,172.09) -- (150.89,86.84) -- (118.46,86.84) -- cycle ;
		\draw   (495.46,117.09) -- (527.89,117.09) -- (527.89,88.76) -- (495.46,88.76) -- cycle ;
		\draw   (418.46,173.07) -- (450.89,173.07) -- (450.89,87.82) -- (418.46,87.82) -- cycle ;
		\draw   (495.46,172.09) -- (527.89,172.09) -- (527.89,116.74) -- (495.46,116.74) -- cycle ;
		\draw   (341.46,173.09) -- (373.89,173.09) -- (373.89,42.27) -- (341.46,42.27) -- cycle ;
		\draw   (457.47,173.09) -- (489.89,173.09) -- (489.89,115.28) -- (457.47,115.28) -- cycle ;
		\draw   (380.46,76.09) -- (412.89,76.09) -- (412.89,28.59) -- (380.46,28.59) -- cycle ;
		\draw   (457.47,115.28) -- (489.89,115.28) -- (489.89,69.4) -- (457.47,69.4) -- cycle ;
		\draw   (457.47,115.28) -- (489.89,115.28) -- (489.89,173.07) -- (457.47,173.07) -- cycle ;
		\draw   (457.47,69.4) -- (489.89,69.4) -- (489.89,115.28) -- (457.47,115.28) -- cycle ;
		\draw   (457.47,40.5) -- (489.89,40.5) -- (489.89,69.39) -- (457.47,69.39) -- cycle ;
		\draw   (418.46,87.84) -- (450.89,87.84) -- (450.89,173.07) -- (418.46,173.07) -- cycle ;
		\draw   (533.46,100.88) -- (565.89,100.88) -- (565.89,129.77) -- (533.46,129.77) -- cycle ;
		\draw   (533.46,72.22) -- (565.89,72.22) -- (565.89,100.88) -- (533.46,100.88) -- cycle ;
		\draw   (380.46,75.76) -- (412.89,75.76) -- (412.89,173.07) -- (380.46,173.07) -- cycle ;
		\draw   (533.46,129.77) -- (565.89,129.77) -- (565.89,172.07) -- (533.46,172.07) -- cycle ;
		\draw   (533.46,72.22) -- (565.89,72.22) -- (565.89,45.51) -- (533.46,45.51) -- cycle ;
		\draw   (457.47,28.62) -- (489.89,28.62) -- (489.89,173.09) -- (457.47,173.09) -- cycle ;
		\draw   (418.46,28.59) -- (450.89,28.59) -- (450.89,173.07) -- (418.46,173.07) -- cycle ;
		\draw   (533.46,27.59) -- (565.89,27.59) -- (565.89,172.07) -- (533.46,172.07) -- cycle ;
		\draw   (380.46,28.59) -- (412.89,28.59) -- (412.89,173.07) -- (380.46,173.07) -- cycle ;
		\draw   (418.46,87.84) -- (450.89,87.84) -- (450.89,61.82) -- (418.46,61.82) -- cycle ;
		\draw   (495.46,116.74) -- (527.89,116.74) -- (527.89,172.07) -- (495.46,172.07) -- cycle ;
		\draw   (495.46,27.59) -- (527.89,27.59) -- (527.89,172.07) -- (495.46,172.07) -- cycle ;
		\draw   (495.46,88.76) -- (527.89,88.76) -- (527.89,116.74) -- (495.46,116.74) -- cycle ;
		\draw   (495.46,60.77) -- (527.89,60.77) -- (527.89,88.76) -- (495.46,88.76) -- cycle ;
		\draw   (572.46,110.2) -- (604.89,110.2) -- (604.89,139.09) -- (572.46,139.09) -- cycle ;
		\draw   (572.46,81.54) -- (604.89,81.54) -- (604.89,110.2) -- (572.46,110.2) -- cycle ;
		\draw   (572.46,139.09) -- (604.89,139.09) -- (604.89,172.79) -- (572.46,172.79) -- cycle ;
		\draw   (572.46,28.32) -- (604.89,28.32) -- (604.89,172.79) -- (572.46,172.79) -- cycle ;
		\draw   (418.46,46.51) -- (450.89,46.51) -- (450.89,61.82) -- (418.46,61.82) -- cycle ;
		\draw   (341.46,28.59) -- (373.89,28.59) -- (373.89,173.07) -- (341.46,173.07) -- cycle ;
		\draw   (341.46,28.59) -- (373.89,28.59) -- (373.89,42.27) -- (341.46,42.27) -- cycle ;
		\draw  [color={rgb, 255:red, 200; green, 200; blue, 200 }  ,draw opacity=1 ] (342,197) .. controls (342.03,201.67) and (344.37,203.99) .. (349.04,203.96) -- (463.06,203.32) .. controls (469.73,203.29) and (473.07,205.6) .. (473.1,210.27) .. controls (473.07,205.6) and (476.39,203.25) .. (483.06,203.21)(480.06,203.23) -- (597.08,202.58) .. controls (601.75,202.55) and (604.07,200.21) .. (604.05,195.54) ;
		\draw   (380.46,75.76) -- (412.89,75.76) -- (412.89,28.27) -- (380.46,28.27) -- cycle ;
		\draw   (418.46,46.51) -- (450.89,46.51) -- (450.89,61.82) -- (418.46,61.82) -- cycle ;
		\draw   (457.47,173.09) -- (489.89,173.09) -- (489.89,115.28) -- (457.47,115.28) -- cycle ;
		\draw   (341.46,173.07) -- (373.89,173.07) -- (373.89,42.25) -- (341.46,42.25) -- cycle ;
		\draw  [fill={rgb, 255:red, 74; green, 144; blue, 226 }  ,fill opacity=0.2 ] (418.46,173.09) -- (450.89,173.09) -- (450.89,87.84) -- (418.46,87.84) -- cycle ;
		\draw  [fill={rgb, 255:red, 74; green, 144; blue, 226 }  ,fill opacity=0.2 ] (341.46,173.09) -- (373.89,173.09) -- (373.89,42.27) -- (341.46,42.27) -- cycle ;
		\draw   (678.46,28.59) -- (710.89,28.59) -- (710.89,173.07) -- (678.46,173.07) -- cycle ;
		\draw  [fill={rgb, 255:red, 74; green, 144; blue, 226 }  ,fill opacity=0.2 ] (635.89,171.58) -- (668.31,171.58) -- (668.31,28.32) -- (635.89,28.32) -- cycle ;
		\draw  [fill={rgb, 255:red, 74; green, 144; blue, 226 }  ,fill opacity=0.2 ] (678.46,173.09) -- (710.89,173.09) -- (710.89,94.83) -- (678.46,94.83) -- cycle ;
		
		\draw (173.68,143.17) node   [align=left] {$ $4 $ $};
		\draw (173.68,91.34) node   [align=left] {$ $3$ $};
		\draw (133.55,129.61) node   [align=left] {$ $6};
		\draw (249.67,85.55) node   [align=left] {2};
		\draw (249.67,114.33) node   [align=left] {$ $2$ $};
		\draw (96.67,123.41) node   [align=left] {$ $7$ $};
		\draw (249.67,149.92) node   [align=left] {$ $3$ $};
		\draw (94.67,49) node   [align=left] {$ $3$ $};
		\draw (173.68,53.94) node   [align=left] {$ $2$ $};
		\draw (134.68,73.83) node   [align=left] {$ $2$ $};
		\draw (249.67,57.86) node   [align=left] {$ $2$ $};
		\draw (134.68,53.16) node   [align=left] {$ $1$ $};
		\draw (211.67,143.4) node   [align=left] {4};
		\draw (211.67,42.91) node   [align=left] {$ $2$ $};
		\draw (211.67,101.75) node   [align=left] {2};
		\draw (211.67,73.76) node   [align=left] {$ $2$ $};
		\draw (288.67,123.65) node   [align=left] {$ $2$ $};
		\draw (288.67,154.94) node   [align=left] {\begin{minipage}[lt]{15.31pt}\setlength\topsep{0pt}
				\begin{flushright}
					$ $2$ $
				\end{flushright}
				
		\end{minipage}};
		\draw (288.67,94.87) node   [align=left] {$ $2$ $};
		\draw (57.67,34.43) node   [align=left] {$ $1$ $};
		\draw (57.67,99.83) node   [align=left] {9};
		\draw (134.68,36.16) node   [align=left] {$ $1$ $};
		\draw (164,205) node [anchor=north west][inner sep=0.75pt]  [color={rgb, 255:red, 155; green, 155; blue, 155 }  ,opacity=1 ] [align=left] {$\displaystyle P_{1}$};
		\draw (58.17,183) node   [align=left] {$\displaystyle B_{1}$};
		\draw (94.52,183) node   [align=left] {$\displaystyle B_{2}$};
		\draw (134.52,183) node   [align=left] {$\displaystyle B_{3}$};
		\draw (172.52,183) node   [align=left] {$\displaystyle B_{4}$};
		\draw (251.52,181.73) node   [align=left] {$\displaystyle B_{6}$};
		\draw (212.52,181.73) node   [align=left] {$\displaystyle B_{5}$};
		\draw (290.52,182.73) node   [align=left] {$\displaystyle B_{7}$};
		\draw (473.68,144.17) node   [align=left] {$ $4 $ $};
		\draw (473.68,92.34) node   [align=left] {$ $3$ $};
		\draw (433.55,130.61) node   [align=left] {$ $6};
		\draw (549.67,86.55) node   [align=left] {2};
		\draw (549.67,115.33) node   [align=left] {$ $2$ $};
		\draw (396.67,124.41) node   [align=left] {$ $7$ $};
		\draw (549.67,150.92) node   [align=left] {$ $3$ $};
		\draw (394.67,50) node   [align=left] {$ $3$ $};
		\draw (473.68,54.94) node   [align=left] {$ $2$ $};
		\draw (434.68,74.83) node   [align=left] {$ $2$ $};
		\draw (549.67,58.86) node   [align=left] {$ $2$ $};
		\draw (434.68,54.16) node   [align=left] {$ $1$ $};
		\draw (511.67,144.4) node   [align=left] {4};
		\draw (511.67,43.91) node   [align=left] {$ $2$ $};
		\draw (511.67,102.75) node   [align=left] {2};
		\draw (511.67,74.76) node   [align=left] {$ $2$ $};
		\draw (588.67,124.65) node   [align=left] {$ $2$ $};
		\draw (588.67,155.94) node   [align=left] {\begin{minipage}[lt]{15.31pt}\setlength\topsep{0pt}
				\begin{flushright}
					$ $2$ $
				\end{flushright}
				
		\end{minipage}};
		\draw (588.67,95.87) node   [align=left] {$ $2$ $};
		\draw (357.67,35.43) node   [align=left] {$ $1$ $};
		\draw (357.67,100.83) node   [align=left] {9};
		\draw (434.68,37.16) node   [align=left] {$ $1$ $};
		\draw (463,206) node [anchor=north west][inner sep=0.75pt]  [color={rgb, 255:red, 155; green, 155; blue, 155 }  ,opacity=1 ] [align=left] {$\displaystyle P_{2}$};
		\draw (358.17,184) node   [align=left] {$\displaystyle B_{1}$};
		\draw (394.52,184) node   [align=left] {$\displaystyle B_{2}$};
		\draw (434.52,184) node   [align=left] {$\displaystyle B_{3}$};
		\draw (472.52,184) node   [align=left] {$\displaystyle B_{4}$};
		\draw (551.52,182.73) node   [align=left] {$\displaystyle B_{6}$};
		\draw (512.52,182.73) node   [align=left] {$\displaystyle B_{5}$};
		\draw (590.52,183.73) node   [align=left] {$\displaystyle B_{7}$};
		\draw (649.67,100.83) node   [align=left] {10};
		\draw (692.67,132.83) node   [align=left] {5};

	\end{tikzpicture}
}

\usepackage{times}

\usepackage{soul}
\usepackage{url}
\usepackage[utf8]{inputenc}
\usepackage{caption}
\usepackage{graphicx}
\usepackage{amsmath}

\newcommand{\pp}{\textsc{ProfilePacking}\xspace}

\newcommand{\opt}{\textsc{Opt}\xspace}
\newcommand{\nextfit}{\textsc{NextFit}\xspace}
\newcommand{\firstfit}{\textsc{FirstFit}\xspace}
\newcommand{\firstfitdec}{\textsc{FirstFitDecreasing}\xspace}
\newcommand{\bestfit}{\textsc{BestFit}\xspace}

\newcommand{\hybrid}{\textsc{Hybrid}($\lambda$)\xspace}
\newcommand{\haware}{\textsc{$H$-Aware}\xspace}
\newcommand{\cnt}{${\tt count}(x)$\xspace}
\newcommand{\ppcnt}{${\tt ppcount}(x)$\xspace}
\newcommand{\adaptive}{\textsc{Adaptive}($w$)\xspace}
\newcommand{\remove}[1]{}

\usepackage{bm}
\usepackage{xspace}
\usepackage{xr}
\usepackage[switch]{lineno}
\usepackage{amsmath,amssymb,amsthm,xspace}
\usepackage{mathtools}
\usepackage{xcolor}

\usepackage{xcolor}         
\usepackage{caption}
\usepackage{subcaption}

\begin{document}

\title{Online Bin Packing with Predictions\thanks{A preliminary version of this work appeared in the Proceedings of the 31st International Joint Conference on Artificial Intelligence~\cite{ijcaiKS22}.}}

\author{%
	Spyros Angelopoulos\thanks{Sorbonne Université, CNRS, LIP6, Paris, France}
	\and
	Shahin Kamali\thanks{Department of Electrical Engineering and Computer Science, York University, Toronto, Canada}
	\and
Kimia Shadkami \thanks{Department of Computer Science, University of Manitoba, Winnipeg, Canada}}


\date{}

\maketitle

\begin{abstract}
	Bin packing is a classic optimization problem with a wide range of applications, from load balancing to supply chain management. In this work, we study the online variant of the problem, in which a sequence of items of various sizes must be placed into a minimum number of bins of uniform capacity. The online algorithm is enhanced with a (potentially erroneous) {\em prediction} concerning the frequency of item sizes in the sequence. We design and analyze online algorithms with efficient tradeoffs between the consistency (i.e., the competitive ratio assuming no prediction error) and the robustness (i.e., the competitive ratio under adversarial error), and whose performance degrades near-optimally as a function of the prediction error. This is the first theoretical and experimental study of online bin packing under competitive analysis, in the realistic setting of learnable predictions. Previous work addressed only extreme cases with respect to the prediction error, and relied on overly powerful and error-free oracles. 
\end{abstract}

\section{Introduction}
\label{sec:introduction}
Bin packing is a classic optimization problem and one of the original NP-hard problems~\cite{garey1979computers}. 
Given a set of {\em items}, each with a (positive) {\em size}, and a bin {\em capacity}, 
the objective is to assign the items to the minimum number of bins so that the sum of item sizes in each bin does not exceed the bin capacity. Bin packing is instrumental in modelling resource allocation problems such as load balancing and scheduling~\cite{10.5555/241938.241940}, and has many applications in areas such as supply chain management, e.g., capacity planning in logistics and cutting stock, but also in cloud computing, where a cloud provider must decide how many physical machines are needed in order to accommodate the incoming jobs~\cite{cohen2019overcommitment}. Several approximation algorithms have been proposed, e.g.,~\cite{DBLP:journals/combinatorica/VegaL81,rothvoss2013approximating,HobergR17}, and efficient exact heuristics motivated by AI settings are often used in practice~\cite{Korf02,Korf03,fukunaga2007bin,SchreiberK13}.

In this work, we focus on the {\em online} variant of bin packing, in which the set of items is not known in advance but is rather revealed in the form of a {\em sequence}. Upon the arrival of a new item, the online algorithm must either place it into one of the currently open bins, as long as this action does not violate the bin's capacity, or into a new bin. The online model has several applications related to dynamic resource management, such as virtual machine placement for server consolidation~\cite{song2013adaptive,wang2011consolidating} and memory allocation in data centers~\cite{bein2011cloud}. Online bin packing has a long history of study; in Section~\ref{subsec:related} we discuss, in more detail, some of the most significant known results in this setting.

In order to analyze the performance of an online algorithm, we will rely on the well-established framework of {\em competitive analysis}, which provides strict, theoretical performance guarantees against worst-case scenarios.  In fact, as stated in~\cite{10.5555/241938.241940}, bin packing has served as ``an early proving ground for this type of analysis in the broader context of online computation''. 
The {\em competitive ratio} of an online algorithm is defined as the worst-case ratio of the algorithm's cost (total number of opened bins) over the optimal offline cost (optimal number of opened bins given knowledge of all items). For bin packing, 
in particular, the standard performance metric is the {\em asymptotic competitive ratio},  in which the optimal offline cost is arbitrarily large~\cite{10.5555/241938.241940}.

While the standard online framework assumes that the algorithm has no information on the input sequence, a recently emerged and very active direction in Machine Learning seeks to leverage {\em predictions} on the input. More precisely, the algorithm has access to some machine-learned information on the input, which, however, may be erroneous; namely, there is a {\em prediction error} $\eta$ associated with it. The objective is to design algorithms which perform well if the prediction is accurate, maintain an efficient competitive ratio is the prediction is highly erroneous (i.e., adversarial), and also exhibit a gentle degradation of the competitive performance, as a function of the prediction error. 
Following the influential work~\cite{DBLP:conf/icml/LykourisV18}, we refer to the competitive ratio of an algorithm with an error-free prediction as the {\em consistency} of the algorithm, and to the competitive ratio with an adversarial prediction as its {\em robustness}. Several online optimization problems have been studied in this learning-augmented setting, including caching~\cite{DBLP:conf/icml/LykourisV18,rohatgi2020near}, ski rental and non-clairvoyant scheduling~\cite{NIPS2018_8174,WeiZ20}, makespan scheduling~\cite{lattanzi2020online}, 
rent-or-buy problems~\cite{DBLP:conf/nips/Banerjee20,anand2020customizing,gollapudi2019online}, 
secretary and matching problems~\cite{DBLP:conf/nips/AntoniadisGKK20,abs-2011-11743}, and metrical task systems~\cite{DBLP:conf/icml/AntoniadisCE0S20}. This is only a partial list of some representative results; see also the survey~\cite{mitzenmacher2020algorithms}.

\subsection{Contribution}
\label{sect:contr}

We give the first theoretical and experimental study of online bin packing with machine-learned predictions. Previous work on this problem has assumed ideal and error-free predictions that must be provided by a very powerful oracle, without any learnability considerations, as we discuss in more detail in Section~\ref{subsec:related}. In contrast, our algorithms exploit natural, and PAC-learnable predictions concerning the {\em frequency} at which item sizes occur in the input,
and our analysis incorporates the prediction error into the performance guarantee. As in other AI-motivated works on bin packing, namely~\cite{Korf02,Korf03,fukunaga2007bin,SchreiberK13}, we assume a discrete model in which item sizes are integers in $[1,k]$ for some constant $k$
(see Section~\ref{sec:prelim}). This assumption is not indispensable, and in Section~\ref{sec:frac} we extend the analysis to items that may have fractional sizes.

We first present and analyze an algorithm called \pp, that achieves optimal consistency, and is also efficient if the prediction error is relatively small. The algorithm builds on the concept of a {\em profile set}, which serves as an approximation of the items that are expected to appear in the sequence, given the frequency predictions. This is a natural concept that, perhaps surprisingly, has not been exploited in the long history of competitive analysis of bin packing, and which can be readily applicable to other online packing problems, such as multi-dimensional packing~\cite{ChristensenKPT17} and vector packing~\cite{azar2013tight}, as we discuss in Section~\ref{sec:conclusion}. 

As the prediction error grows, \pp may not be robust; we show, however, that this is an unavoidable price that any optimally-consistent algorithm with frequency predictions must pay. We thus design and analyze a more general class of {\em hybrid} algorithms that combine \pp and any one of the known robust online algorithms, and which offers a more balanced theoretical tradeoff between robustness and consistency. 

We perform extensive experiments on our algorithms. Specifically, we evaluate them on a variety of 
publicly available benchmarks, such as the BPPLIB benchmarks~\cite{benchmarks}, but also on distributions studied specifically in the context of offline bin packing, such as the {\em Weibull} distribution~\cite{CastineirasCO12}. The results show that our algorithms outperform the known efficient algorithms without any predictions. We also evaluate a heuristic that updates the predictions based on previously served items, and which is better suited for inputs generated from distributions that change over time (e.g., in the case of {\em evolving data}~\cite{gomes2017survey}).


Last, we show that our algorithms can be applicable in other settings. Specifically, we show an application of our algorithms in the context of {\em Virtual Machine} (VM) placement in large data centers~\cite{Mann15}: here, we obtain a more refined competitive analysis in terms of the {\em consolidation ratio}, which reflects the maximum number of VMs per physical machine, in typical scenarios.  Furthermore, we show that our analysis of \pp has a direct application in the {\em sampling-based} setting, in which the algorithm can access a small sample of the input, and the objective is to obtain an online algorithm that performs efficiently as a function of the number of sampled input items. Thus, our online algorithms can also serve as fast approximations to the offline problem, since frequency prediction with bounded error can be attained with a small sample size.

In terms of analysis techniques, we note that the theoretical analysis of the algorithms we present is specific to the setting at hand and treats items ``collectively''. In contrast, almost all known online bin packing algorithms are analyzed using a {\em weighting} technique~\cite{10.5555/241938.241940}, which treats each bin ``individually'' and independently from the others (by assigning weights to items and independently comparing a bin’s weight in the online algorithm and the optimal offline solution). In terms of the experimental analysis, in our experiments, the prediction error is a natural byproduct of the learning phase, and predictions are obtained by observing a small prefix of the input sequence. This is in contrast to several works in learning-enhanced algorithms, in which a perfect prediction is first generated by a very powerful oracle, then some random error is applied in order to simulate the imperfect prediction.
\subsection{Related Work} 
\label{subsec:related}
Online bin packing has a long history of study. The simplest algorithm is \nextfit, which places an item into its single open bin when possible; otherwise, it closes the bin (does not use it anymore) and opens a new bin for the item. \firstfit is another simple heuristic that places an item into the first bin of sufficient space and opens a new bin if required. \bestfit works similarly, except that it places the item into the bin of minimum available capacity, which can still fit the item. \nextfit has a competitive ratio of 2, while both \firstfit and \bestfit are 1.7-competitive~\cite{10.5555/241938.241940,JoDUGG74}. {Improving upon this performance requires more sophisticated algorithms, and many have been proposed in the literature.} 
These algorithms are variants of the classic {\em Harmonic} algorithm~\cite{lee1985simple}, which places items of approximately equal sizes, according to a harmonic sequence, in the same bin.
The currently best algorithm is the Advanced Harmonic (AH) algorithm, which has a competitive ratio of 1.57829~\cite{BaloghBDEL18}, whereas the best-known lower bound on the competitive ratio is 1.54278~\cite{BaloghBDELNEW21}. However, one should note that the Harmonic family of algorithms are designed specifically for improving the competitive ratio, and their typical performance is inferior to heuristics that are widely used in practice such as \firstfit and \bestfit, as shown in~\cite{KamaliL15}. 

Beyond competitive analysis, the bin packing problem has been studied under stochastic settings, in which the sizes of items are independent and identically distributed samples from
some distribution~\cite{Csir06,RheeT93,GuptaR20}. 
In this setting, the objective is to minimize the expected \emph{loss}, defined as the difference between the number of bins opened by the algorithm, and the total size of all items normalized by the bin capacity. 
Ideally, one aims for a loss that is as small as $o(n)$, where $n$ is the number of items in the sequence. For example, the Sum-of-Squares (SS) algorithm of~\cite{Csir06} has an expected loss of $O(\sqrt{n})$. The same guarantee can be achieved by an algorithm whose actions only depend on the size of the arriving item, and the levels of the bins in the current packing~\cite{GuptaR20}. The algorithms in the above works are oblivious to the distribution; however, if the length of the input $n$ and the distribution are known to the algorithm, the expected loss can be reduced to $O(1)$~\cite{Banerjee020}. We note, however, that all these algorithms are designed for stochastic inputs, and do not provide efficient worst-case guarantees on the competitive ratio. For example, the Sum-of-Square algorithm has a competitive ratio in the range $[2,2.\bar{7}]$~\cite{Csir06}, which is as bad as (and possibly worse than) the naive \nextfit algorithm. 	

Online bin packing has also been studied under the {\em advice complexity} model~\cite{BoyarKLL16,Mikkelsen16,ADKRR18}, in which the online algorithm has access to some error-free information on the input called {\em advice}. The objective is to quantify the tradeoffs between the competitive ratio and the size of the advice (i.e., the number of bits in the binary encoding of the advice). For instance,~\cite{ADKRR18} showed that $O(1)$ advice bits suffice to improve the competitive ratio of the problem to $1.5+\epsilon$, for any constant $\epsilon>0$. We emphasize, however, that such studies are only of theoretical interest for two reasons: First, the advice is assumed to have no errors, and it is possible that a single bit flip gravely affects the competitive ratio; Second, the advice is assumed to encode  {\em any} information, without any learnability considerations, and it may thus be provided by an omnipotent oracle that knows the optimal solution. To illustrate with an example, typical advice in the above works encodes information about the number of bins in the optimal offline packing that contain relatively large items (of size at least a third of the bin capacity) or relatively small items (less than a third of the bin capacity). 

Online bin packing was recently studied under an extension of the advice complexity model, in which the advice may be {\em untrusted}~\cite{DBLP:conf/innovations/0001DJKR20}. Here, the algorithm's performance is evaluated only at the extreme cases in which the advice is either error-free or adversarially generated, namely with respect to its consistency and its robustness, respectively. The objective is to find Pareto-efficient algorithms concerning these two metrics, as a function of the advice size. However, this model is not concerned with the algorithm's performance in typical cases in which the prediction does not fall in one of the two above extremes, does not incorporate the prediction error into the analysis, and does not consider the learnability aspects of the advice. In particular, even with error-free predictions, the algorithm of~\cite{DBLP:conf/innovations/0001DJKR20} has a competitive ratio as large as 1.5, whereas a single bit error may result in a competitive ratio that is as large as 6. 

Concerning the application of frequency predictions in competitive online optimization, we note that, perhaps surprisingly, such predictions have not been used widely, despite their simplicity and effectiveness.  \cite{MahdianNS07} gave improved competitive ratios for a  generalized online matching problem motivated by advertisement space allocation, using unreliable frequency estimates of keywords.  Recently, and concurrently with the conference version of our work, \cite{knapsack22} studied the online knapsack problem under frequency predictions, where each item has a value and a size, and the objective is to maximize the value of items that are accepted (and can fit) in the knapsack. Here, the concept of ``frequency prediction'' has a more liberal notion: it describes, for each possible value, an estimate of the total size of all items of that value. In contrast, in the bin packing setting, item frequency is a less complicated concept, which benefits the design and applicability of the corresponding algorithms.



\section{Online Bin Packing: Model and Predictions} 
\label{sec:prelim}
We begin with some preliminary discussions related to online bin packing. 
The input to the online algorithm is a sequence  $\sigma= a_1, \ldots ,a_n$, where $a_i$ is the size of the $i$-th item in $\sigma$. We denote by $n$ the length of $\sigma$, and by $\sigma[i,j]$ the subsequence of $\sigma$ that consists of items with indices $i, \ldots, j$ in $\sigma$.

We denote by $k \in \mathbb{Z}^+$ the bin capacity. Note that $k$ is independent of $n$, and is thus constant. 
We assume that the size of each item is an integer in $[1,k]$, where $k$ is the bin capacity. This is a natural assumption on which many efficient algorithms for bin packing rely, e.g.,~\cite{SchreiberK13,fukunaga2007bin,csirik2006sum}.
Furthermore, 
without any restriction on the item sizes, ~\cite{Mikkelsen16} showed that no online algorithm with advice of size sublinear in the size of the input can have competitive ratio better than 1.17 (even if the advice is error-free). This negative result implies that some restriction on item sizes is required so as to leverage frequency predictions. 
Nevertheless, in Section~\ref{sec:frac} we show how to handle fractional items, and how to express the performance loss due to such items.


Given an online algorithm $A$ (with no predictions), we denote by $A(\sigma)$ its output on input $\sigma$, i.e., the packing it produces, and by $|A(\sigma)|$ the number of bins in its output.
We denote by $\opt(\sigma)$ the offline optimal algorithm with knowledge of the input sequence. The (asymptotic) competitive ratio of $A$ is defined~\cite{10.5555/241938.241940} as
$$
\lim_{n\to \infty} \sup_{\sigma: |\sigma|=n} \frac{|{A(\sigma)}|}{{|\opt(\sigma)|}}.
$$

Consider a bin $b$. For the purpose of the analysis, we will often associate $b$ with a specific configuration of items that can be placed into it. We thus say that $b$ is of 
{\em type} $(s_1, s_2, \ldots ,s_l, e)$, with $s_i\in[1,k]$, $e \in [0,k]$ and $\sum_{j=1}^l s_j +e=k$, in the sense that the bin
can pack $l$ items of sizes $s_1, \ldots ,s_l$, with a remaining empty space equal to $e$. We specify that a bin is {\em filled according to type} $(s_1, s_2, \ldots ,s_l, e)$, if it contains $l$ items of sizes $s_1, \ldots ,s_l$, with an empty space $e$. Note that a type induces a partition of $k$ into $l+1$ integers; we call each of the $l$ elements $s_1, \ldots ,s_l$ a {\em placeholder}, and
denote by $\tau_k$ the number of all possible bin types. Observe that $\tau_k$ depends only on $k$ and not on the length $n$ of the sequence, and is constant, since $k$ is constant. 


Consider an input sequence $\sigma$. For any $x \in [1,k]$, let $n_{x,\sigma}$ 
denote the number of items of size $x$ in $\sigma$. We define the {\em frequency of size $x$ in $\sigma$}, denoted by $f_{x,\sigma}$, to be equal to $n_{x,\sigma}/n$, hence $f_{x,\sigma} \in [0,1]$, 
and $\sum_{x \in [1,k]} f_{x,\sigma}=1$. 
Our algorithms will use size frequencies as predictions. Namely, for every $x \in [1,k]$, there is a predicted value of the frequency of size $x$ in $\sigma$, which we denote by $f'_{x,\sigma}$. The predictions come with an error, and in general, $f'_{x,\sigma} \neq f_{x,\sigma}$; for instance, it may very well be the case that $\sum_{x \in [1,k]} f'_{x,\sigma} \neq 1$.
To quantify the prediction error, let $\bm{f_\sigma}$ and $\bm{f'_\sigma}$ denote the frequencies and their predictions in $\sigma$, respectively, as points in the $k$-dimensional space. 
In line with previous work on online algorithms with predictions, e.g.~\cite{NIPS2018_8174}, we can define 
the error $\eta$ as the $L_1$ norm of the distance between $\bm{f_\sigma}$ and $\bm{f'_\sigma}$. 
These size frequencies are PAC-learnable, due to the following (folklore) fact:

\begin{remark}\cite{canonne2020short}.
\label{remark:pac}
for any given  $\epsilon>0$ and $\delta \in (0,1]$, a sample of size $\Theta((k+ \log (1/\delta))/{\epsilon ^2})$ is sufficient (and necessary) to learn the frequencies of $k$ item sizes with accuracy  $\epsilon$ and error probability $\delta$, assuming the distance measure is the $L_1$-distance.
\end{remark}



We denote by $A(\sigma,\bm{f'_\sigma})$ the output of $A$ on input $\sigma$ and
predictions $\bm{f'}_\sigma$. 
To simplify notation, we will omit $\sigma$ when it is clear from context, i.e., we will use $\bm{f'}$ in place of $\bm{f'_\sigma}$.


\section{Profile Packing}
\label{sec:profile}
In this section, we present and analyze an online algorithm with predictions $\bm{f'}$, which we call \pp. The algorithm is based on the concept of a {\em profile}, denoted by $P_{\bm{f'}}$, which we define as the multiset that consists of $\lceil f'_x m \rceil$ items of size $x$, for all $x\in [1,k]$. Here, $m$ is a parameter that is a sufficiently large, but constant integer, which will be specified later. One may thus think of the profile as an ``approximation'' of the multiset of items that is expected as input, given the predictions $\bm{f'}$.

Consider an {\em optimal} packing of the items in $P_{\bm{f'}}$. Since the size of items in 
$P_{\bm{f'}}$ is bounded by $k$, it is possible to compute such an optimal packing in constant 
time, e.g., using an efficient exact heuristic~\cite{Korf02}). 
We will denote by $p_{\bm{f'}}$ the number of bins in the optimal packing of all items in the profile. Note that each of these $p_{\bm{f'}}$ bins is filled according to a certain type that is specified by the optimal packing of the profile.  We simplify notation and
use $P$ and $p$ instead of $P_{\bm{f'}}$ and $p_{\bm{f'}}$, respectively, when $\bm{f'}$ is implied.

We define the actions of \pp. Prior to serving any items, \pp opens $p$ empty bins of types that are in accordance with the optimal packing of the profile (so that there are $\lceil f'_x m \rceil$ placeholders of size $x$ in these empty bins). When an item, say of size $x$, arrives, the algorithm will place it into any placeholder reserved for items of size $x$, provided that such one exists. Otherwise, i.e., if all placeholders for size $x$ are occupied, the algorithm will open another set of $p$ bins, again of types determined by the optimal profile packing. We call each such set of $p$ bins a {\em profile group}. 
Note that the algorithm does not close any bins at any time, that is, any placeholder for 
an item of size $x$ can be used at any point in time, so long as it is unoccupied.

We require that \pp opens bins in a \emph{lazy} manner, that is, the $p$ bins in the profile group are opened virtually, and each bin contributes to the cost only after receiving an item.
Last, suppose that for some size $x$, it is $f_x >0$, whereas its prediction is $f'_x=0$. In this case, $x$ is not in the profile set $P$. We call items of such size \emph{special}. 
\pp packs these special items separately from others, using \firstfit. Algorithm~\ref{AlgPP} describes \pp in pseudocode.

{\begin{center}
		\algrenewcomment[1]{\(\triangleright\) #1}
		\scalebox{.95}{
			\begin{minipage}[htb!]{\columnwidth}
				\newpage \ \vspace{-1cm} \\ 
				\begin{algorithm}[H]
					\caption{$\pp$}\label{AlgPP}
					\textbf{Input:} \parbox[t]{\linewidth}{$\sigma$: the input sequence with items in $[1,k]$} 
					{\parbox[t]{\linewidth}{$\hspace{1cm}{\bm f'}$: predicted item frequencies ($\forall x\in [1,k], {f'_x} \in [0,1]$)}}\vspace*{1mm}
					\textbf{Output:} a packing of $\sigma$ (a set of bins that contain all items in $\sigma$) \vspace*{2mm}
					\      \\
					\hspace*{5mm}\Comment{form the profile set.\vspace*{-.8mm}}
					\begin{algorithmic}[1]
						\State{$P_{\bm f'} \leftarrow \phi$ }\label{lineInitStart}  
						\FOR {$x \in \{1, \ldots k \}$}
						\State \hspace*{\algorithmicindent}
						$P_{\bm f'} \leftarrow P_{\bm f'} \cup \{\lceil f'_x m\rceil \text{ items of size } x \}$
						\ENDFOR 
						\ \vspace*{1.4mm}
						\Commentak{compute an optimal packing of the profile set in which each bin is partitioned to placeholders.}
						\State $\opt_{\bm{f'}}(P) = $ optimal packing of $P_{\bm{f'}}$.\label{optPack}
						\State{ $p_{f'} \leftarrow |\opt_{\bm{f'}}(P)|$ }
						\State{\ngr $\ \leftarrow$  $p_{f'}$ empty bins in accordance with $\opt_{\bm{f'}}(P)$.}\hfill \Comment{open the first profile group}
						\State{\eb $\ \leftarrow$ \ngr} \hfill\Comment{bins that are opened but empty}
						\State{\neb $\ \leftarrow \phi$}\hfill\Comment{bins that contribute to the cost}\label{lineInitEnd}    
						\vspace*{2mm}
						
						\FOR{$i \in  (1, \ldots, n)$} \label{line:act1} \hfill \Comment{packing the sequence in an online manner}
						\State{ $x \leftarrow \sigma[i]$}
						\IF{$\opt_{\bm{f'}}(P)$ has no placeholder of size $x$}
						\State{use \firstfit to pack $\sigma[i]$}\hfill \Comment{\ $x$ is a special item}
						\ELSE
						\State $N_x \leftarrow$ bins in \neb \ with placeholder for  $x$  
						\IF{$N_x \neq \phi$}\hfill\Comment{place $\sigma[i]$ in a non-empty bin}
						\State $B \leftarrow$ any bin of $N_x$\label{anyBin1}
						\State place $\sigma[i]$ in a placeholder of size $x$ in $B$
						\ELSE
						\State $E_x \leftarrow $ bins in \eb \ with placeholder for  $x$  
						\IF{$E_x = \phi$}\hfill
						\hspace*{1.3cm}\Comment{open a new profile group}
						\State{\ngr $\ \leftarrow$  $p_{f'}$ new bins as in $\opt_{\bm{f'}}(P)$}
						\State{\eb $\ \leftarrow$ \eb $\ \cup$ \ngr} \ 
						\ENDIF  \\
						\ \ \ \ \hspace*{.75cm}
						\Comment{place $\sigma[i]$ in a (virtually) opened empty bin}
						\State $B \leftarrow$ any bin of $E_x$\label{anyBin2}
						\State place $\sigma[i]$ in a placeholder of size $x$ in $B$
						\State \eb $\ \leftarrow$ \eb $\setminus \  \{B\} $
						\State \neb $\ \leftarrow$ \neb $\ \cup\  \{B\} $
						\ENDIF
						\ENDIF
						\ENDFOR \label{line:act2}
						\ \\ \textbf{return} \neb \hfill\Comment{return the set of non-empty bins} 
						
					\end{algorithmic}
				\end{algorithm}
			\end{minipage}
		}
	\end{center}
}

\newpage
 \begin{example}\label{ex:2}
	Suppose that $k=10$, $m=20$, and that the predictions $\bm{f'}$ are given in the following table.  
	
	\begin{center}
		\scalebox{.9}{
	\begin{tabular}{|c|c|c|c|c|c|c|c|c|c|c|}\hline
	$x$ & 1 & 2 & 3 & 4 & 5 & 6 & 7 & 8 & 9 & 10 \\
	\hline 
$f'_x$  & 0.11 & 0.53 & 0.15 & 0.10 & 0 & 0.03 & 0.05 & 0 & 0.03 & 0 \\
	\hline
	\end{tabular}}
	\end{center}

	%
	%
	The profile based on these frequencies is  $P_{\bm f'} = \{3 \times 1, {11} \times 2,3 \times 3, 2\times 4, 6,7,9\}$, where $i \times x$ indicates $i$ items of size $x$. For example, there are $\lceil f'_1 m\rceil = \lceil 0.11 \cdot 20 \rceil = 3$ items of size $x=1$ in the profile. There is an optimal packing of $P_{\bm f'}$ that consists of 7 bins. Figure~\ref{fig:pp} illustrates this packing, as well as the packing of \pp on an example online sequence.  

\begin{figure}[!h]
	\centering
\hspace*{-3mm}	\scalebox{.9}{\figOne}
	\caption{The packing output by \pp on the input sequence $\sigma = 2,3,1,4,10,2,9,4,6,9,2,6,5$. The predictions and the corresponding profile are as given in Example~\ref{ex:2}. The optimal packing of the profile consists of seven bins $B_1, \ldots, B_7$. When serving $\sigma$, the algorithm opens two profile groups, denoted by $P_1$ and $P_2$. The profile group $P_2$ is opened upon serving the second item of size 9, namely the 10th request in the sequence. The placeholders that have received an item are highlighted. Items 10 and 5 are special items and are packed using \firstfit. The total cost for serving $\sigma$ is equal to 9, and there are 7 bins that are only opened virtually, hence they do not contribute to the cost.
       \label{fig:pp}}		
\end{figure}
\end{example}

\subsection{Analysis of \pp}
\label{subsec:analysis}

We first show that in the ideal setting of error-free prediction, \pp is near-optimal (Lemma~\ref{th:no-error}). This result will be very useful in the analysis of the more realistic setting of erroneous predictions (Theorem~\ref{th:consistent}). We denote by $\epsilon$ any fixed constant less than 0.5, and in order to achieve consistency equal to $1+\epsilon$, it will suffice to define $m$ to be such that $m \geq  3\tau_k k/\epsilon$, as will become evident in the proof of Lemma~\ref{th:no-error}. Given that $k$ (and thus $\tau_k$) and $\epsilon$ are constants, $m$ is also a  constant.

\begin{lemma} \label{th:no-error}
	For any constant $\epsilon\in(0,0.5]$, and error-free prediction ($\bm{f'}=\bm{f}$), \pp has competitive
	ratio at most $1+\epsilon$. 
\end{lemma}

\begin{proof} 
	Given an input sequence $\sigma$,
 denote by $PP(\sigma, \bm{f'})$ the packing output by the algorithm. This output can be seen as consisting of $g$ profile group packings for some positive integer $g$ (since each time the algorithm allocates a new set of $p$ bins, a new profile group is generated). Since the input consists of $n$ items, and the profile has at least $m$ items, we have that $ g \leq \lceil {n}/{m}\rceil $. 
	
	Given an optimal packing $\opt(\sigma)$, we define a new packing, denoted by $N$, that not only packs items in $\sigma$, but also additional items as follows. $N$ contains all (filled) bins of $\opt(\sigma)$, along with their corresponding items. For every bin type in $\opt(\sigma)$, we want that
	$N$ contains a number of bins of that type that is divisible by $g$. To this end, we add up to $g-1$ filled bins of the same type in $N$. 
	
	We can argue that $|N|$ is not much bigger than $|\opt(\sigma)|$. We have that
       \[ 
	|N| \leq |\opt(\sigma)| + (g-1)\tau_k < |\opt(\sigma)| 
	+ n\tau_k/m \leq |\opt(\sigma)| ( 1+ \tau_kk/m),
       \] 
	Specifically, the first inequality holds from the way $N$ was generated, the second inequality holds because $g \leq \lceil n/m \rceil$ (thus $g-1 < n/m$),  and the last inequality holds because $|\opt(\sigma)| \geq \lceil n/k\rceil$ (each bin can contain at most $k$ items).
	Let $\epsilon' = \epsilon/3$ and recall that we have chosen $m$ so that $m \geq  \tau_k k/\epsilon'$. We conclude that 
       \begin{align}
       |N| \leq (1+\epsilon') |\opt(\sigma)|. \label{ineqqq1} 
	\end{align}	
	By construction, $N$ contains $g$ copies of the same bin (i.e., bins that are filled according to the same type). Equivalently, $N$ consists of $g$ copies of the same packing, which we denote by $\overline{N}$. Define $q=|\overline{N}|$ to be the number of bins in this packing.
	We will show that $p$ is not much bigger than $q$, which is crucial in the proof. 
	The number of items of size $x$ in the packing $\overline{N}$
	is at least $\lceil n_x/g \rceil$, since $N$ contains at least $n_x$ items of size $x$. 
	We can give the following lower bound on $\lceil n_x/g \rceil$.
  \begin{align*}
	\lceil n_x/g \rceil   
	&\geq \  n_x/\lceil n/m \rceil  \tag{ 	$g \leq \lceil n/m \rceil$} \\
	&> \ n_x m/(n+m) \tag{$\lceil n/m \rceil < (n+m)/m$} \\
	&=  \  n_x (m/n - m^2/(n^2+mn)) \\ 
	&\geq  \  n_x m/n - m^2/(n+m)  \tag{ $n_x\leq n$} \\ 
	&\geq    \  \lceil n_x m/n \rceil -1 - m^2/(n+m)  
	\tag{$y \geq \lceil y\rceil-1$ for any $y$  }
	\\ 
	&> \  \lceil n_x m/n \rceil - 2 \tag{$m^2<n$}. \\ 
\end{align*}
The last inequality holds because $m$ is a constant with respect to $n$, which defines the input size. 
We conclude, from the above, that $\lceil n_x/g \rceil   \geq  \lceil n_x m/n \rceil - 1$.	
Given that there are
$\lceil f'_x m \rceil = \lceil n_x m/n \rceil$ items of size $x$ in the profile set, we can further conclude that for any $x\in [1,k]$, $\overline{N}$ packs all items of size $x$ that appears in the profile set, with the exception of at most one such item. From the statement of \pp, and its optimal packing of the profile set, we infer that 
\begin{align}
q+k \geq p, 
\label{ineqqq2}
\end{align}
and recall that $p$ denotes the size of the optimal packing of the profile. Moreover, we have:
       \begin{align*}
       	q  = |\overline{N}| & = |N|/g \tag{by definition of $q$}\\
	 &\geq |\opt(\sigma)|/g \geq n/(k g) \tag{$|N| \geq |\opt(\sigma)| \geq n/k$}  \\ 
       & \geq n/(k\lceil n/m\rceil)  \tag{$g \leq \lceil n/m \rceil $}\\
       &> (\lceil n/m \rceil m- m)/(k \lceil n/m \rceil ) \\
       &\geq  m/k -m^2/kn \\
       &> m/k - \epsilon' \tag{$m^2/kn\in o(1)$ and $\epsilon' \in \Theta(1)$ } \\
       & >  \tau_k/\epsilon' -\epsilon' \tag{$m \geq  \tau_k k/\epsilon'$}\\
       & >  (\tau_k-1)/\epsilon' >   k/\epsilon' \tag{$\epsilon' < 1/\epsilon'$, $\tau_k > k+1$}.
       \end{align*}
       We thus showed that $q > k/\epsilon'$, hence~\eqref{ineqqq2} implies that 
	\begin{align}
	p< q(1+\epsilon').
		\label{ineqqq3}
	\end{align}
	We conclude that the number of bins in each profile group is within a factor $(1+\epsilon')$ of the 
	number of bins in $\overline{N}$. Moreover, recall that $PP(\sigma,\bm{f'})$ consists of $g$ profile groups, and $N$ consists of $g$ copies of $\overline{N}$. Combining this with previously shown properties, we have that 
	\begin{align*}
       |PP(\sigma,\bm{f'})| &\leq g \cdot p  \\
       &< g(1+\epsilon')q \tag{by Inequality (\ref{ineqqq3})} \\
       &\leq (1+\epsilon')(1+\epsilon')|\opt(\sigma)| \tag{by Inequality (\ref{ineqqq1})} \\
       &< (1+3\epsilon')|\opt(\sigma)| \tag{$(1+\epsilon')^2 < 1+3\epsilon'$} \\ & = (1+\epsilon)|\opt(\sigma)|,
       \end{align*}
       which concludes the proof.
\end{proof}

We will now use Lemma~\ref{th:no-error} to prove a more general result that incorporates the prediction error into the analysis. To this end, we will relate the cost of the packing of \pp to the packing that the algorithm would output if the prediction were error-free, which will allow us to apply the result of Lemma~\ref{th:no-error}. Specifically, we will argue that in the presence of prediction error, the cost of \pp may be affected in two ways: The number of bins in a single profile of \pp may increase, and more profiles may have to be opened. In the proof of the following theorem, for each of these two cases, we bound the number of additional opened bins as a function of error.

\begin{theorem} \label{th:consistent}
	For any constant $\epsilon\in(0,0.5]$, and predictions $\bm{f'}$ with error $\eta$, \pp has competitive ratio at most $1+(2+5\epsilon)\eta k+\epsilon$.
\end{theorem}

\begin{proof}
	Let $\bm{f}$ be the frequency vector for the input $\sigma$, where  
	$\bm{f}$ is unknown to the algorithm. In this context, $PP(\sigma,\bm{f})$ is the packing output by \pp with error-free prediction, and from Lemma~\ref{th:no-error} we know that
	$
	|PP(\sigma,\bm{f})| \leq (1+\epsilon) |\opt(\sigma)|.
	\label{eq:pp_f}
	$ 
	Recall that $P_{\bm{f'}}$ denotes the profile set of \pp on input $\sigma$ with predictions $\bm{f'}$, and $p_{\bm{f'}}$ denotes the number of bins in the optimal packing of $P_{\bm{f'}}$; $P_{\bm{f}}$ and $p_{\bm{f}}$ are defined similarly. 
	We will first relate $p_{\bm{f}}$ and $p_{\bm{f'}}$ in terms of the error $\eta$. Note that the multisets $P_{\bm{f}}$
	and $P_{\bm{f'}}$ differ in at most $\sum_{x=1}^k \mu_x$ elements, where $\mu_x=|\lceil f_x m \rceil -\lceil f'_x m \rceil |$. We call these elements {\em differing}. We have $\mu_x \leq |(f_x-f'_x)m|+1$, hence $\sum_{x=1}^k \mu_x \leq k+ \sum_{x=1}^k |(f_x-f'_x)m|\leq k+ \eta m$, where 
	the last inequality holds because $\eta$ is the $L_1$ norm of the distance between $\bm{f_\sigma}$ and $\bm{f'_\sigma}$, that is $\eta =\sum_{x=1}^k |(f_x-f'_x)|$. 
	We conclude that the number of bins in the optimal packing of $P_{\bm{f'}}$ exceeds the number of bins in the optimal packing of $P_{\bm{f}}$ by at most $k+\eta m$, i.e., $p_{\bm{f'}} \leq p_{\bm f}+k +\eta m$.
	
	Let $g$ and $g'$ denote the number of profile groups in $PP(\sigma,\bm{f})$ and $PP(\sigma,\bm{f'})$, respectively.
	We aim to bound $|PP(\sigma,\bm{f'})|$, and to this end we will first show a bound on the number of bins opened by $PP(\sigma,\bm{f'})$ in its first $g$ profile groups, then in on the number of bins in its remaining $g'-g$ profile groups (if $g'\leq g$, there is no such contribution to the total cost). For the first part, the bound follows easily: There are $g$ profile groups, each one consisting of $p_{\bm{f'}}$ bins, therefore the number of bins in question is at most 
	$g \cdot p_{\bm{f'}} \leq g (p_{\bm{f}}+k+\eta m)$. For the second part, since \pp is lazy, any item packed by $PP(\sigma,\bm{f'})$ in its last $g'-g$ packings has to be a differing element, which implies from the discussion above that $PP(\sigma,\bm{f'})$ opens at most $g(k+\eta m)$ bins in its last $g'-g$ profile groups. 
	The result follows then from the following inequalities: 
\begin{align*}
	& |PP(\sigma, \bm{f'})|   \\ 
	&\leq  g(p_{\bm {f}}+k+\eta m) +g (k+\eta m)   \\
	&    =  \  g(p_{\bm {f}}+2k+2\eta m)  \ \\
	& \leq  \  g(p_{\bm {f}}+ 2\eta m(1+\epsilon)) & \tag{$k \leq \epsilon m$}  \\
	& \leq  \  g(p_{\bm {f}}+ 2\eta \ p_{\bm {f}} k(1+\epsilon)) & 
	\tag{$p_{\bm {f}} \geq \lceil m/k\rceil$}  \label{Line:critical} \\
	& = g\cdot p_{\bm{f}}(1+2\eta k(1+\epsilon))  \label{Line:subseq1} \\
	& \leq  |PP(\sigma,\bm{f})|(1+2\eta k(1+\epsilon)) &  \\
	& \leq (1+\epsilon)(1+2\eta k(1+\epsilon)) |\opt(\sigma)| &  \tag{since   $|PP(\sigma,\bm{f})| \leq (1+\epsilon) |\opt(\sigma)|$}\\
	& = (1+2\eta k(1+\epsilon)^2+\epsilon) |\opt(\sigma)| & \\
	& = (1+2\eta k(1+\epsilon^2 + 2\epsilon)+\epsilon) |\opt(\sigma)| \\
	& < (1+ \eta k(2+ 5\epsilon)+\epsilon) |\opt(\sigma)|. & \tag{$\epsilon^2 < \epsilon/2$} \label{Line:subseq2}     
\end{align*}

\end{proof}


Theorem~\ref{th:consistent} shows that \pp has robustness that is linear in $k$.
 The following result proves that this is an unavoidable price that {\em any} online algorithm with frequency-based predictions must pay.

\begin{theorem}
	Let $c$ be any constant with $c<1$. Then for any $\alpha\leq c/k$,
	any algorithm with frequency predictions that is $(1+\alpha)$-consistent must have robustness at least $(1-c)k/2$. \label{th:tradeoff}
\end{theorem}

\begin{proof}
		For simplicity, we can assume that $n$ is divisible by $k$.
	Suppose that the prediction indicates that in the input sequence $\sigma$, half of the items have size $1$, while the remaining items have size $k-1$, i.e., we have $f'_{x,\sigma}=1/2$, if $x\in \{1,k-1\}$, and $f'_{x,\sigma}=0$, otherwise.
	Define $\sigma_1$ as the sequence that consists of $n$ items of size $1$ followed by $n$ items of size $k-1$, and $\sigma_2$ as the sequence that consists of $n$ items of size $1$ followed by $n$ items of size $1$. 

	Suppose first that the input is $\sigma_1$, then the above-defined prediction is error-free. 
	Moreover, $\opt(\sigma_1) = n$ and hence for an algorithm $A$ to be $(1+\alpha)$-consistent, it must open at most $(1+\alpha)n + o(n)$ bins. Out of these bins, $n$ bins each receive an item of size $k-1$ and possibly an item of size 1.
	Each of the remaining $\alpha n + o(n)$ bins may contain up to $k$ items of size 1; that is, they contain at most $k \alpha n + o(n)$ items of size 1. Therefore, $n- k \alpha n - o(n)$ items of size 1 must each be placed together with items of size $k-1$.	
This implies that $A$ opens at least $(1-k \alpha)n -o(n)$ bins when serving the first $n$ items of size 1. 
	
	Next, suppose that the input is $\sigma_2$. 
	We have ${f}_{1,\sigma_2} = 1$ and ${f}_{k-1,\sigma_2} = 0$, and consequently the error is $\eta=1$. 
		The optimal packing is formed by $2n/k$ bins, each containing $k$ items of size 1; that is, $\opt(\sigma_2) =  2n/k$. On the other hand, by the above observation, the cost of $A$ is at least $(1-k \alpha)n -o(n)$. Hence, the robustness of $A$ is at least $\frac{(1-k \alpha)n -o(n)} {2n/k}$. Given that $\alpha \leq c/k$, it follows that $A$ has robustness at least $(1-c)k/2$.
\end{proof}

We conclude that the robustness of \pp is close-to-optimal and no $(1+\epsilon)$-consistent algorithm can do asymptotically better. It is possible, however, to obtain more general tradeoffs between consistency and robustness, as we discuss in the next section.

\paragraph{Time complexity of \pp}
We bound the overall time complexity of \pp for serving a sequence of $n$ items as a function of $n, k$, and $m$. The initial phase of the algorithm, which involves computing the profile and its optimal packing, runs in time independent of $n$ and does not impact the asymptotic time complexity. It is possible to find the optimal packing of the profile set using efficient exact heuristics such as~\cite {fukunaga2007bin,SchreiberK13}. If faster pre-processing is required, one can replace the exact optimal packing with an approximate packing using simple heuristics like \firstfitdec~\cite{dosa2007tight}, which has a competitive ratio of $11/9$. This will improve the empirical running time, while increasing the number of opened bins by the same ratio. Such an approach is also useful in settings where predictions are updated based on previously served items, and thus the packing of the profile set must be computed periodically. Overall, the worst-case time complexity of \pp is $O(kmn)$. Note that each item is served in amortized time $O(km)$, which is constant since $k$ and $m$ are constants. 

\section{A Broader Class of Algorithms}
\label{sec:hybrid}

In this section, we describe and analyze a more general class of algorithms which offer better robustness in comparison to \pp, at the expense of slightly worse consistency. To this end, we will combine \pp with any algorithm $A$ that has efficient worst-case competitive ratio, in the 
standard online model in which there is no prediction. 
Specifically, we will define a class of algorithms based on a parameter $\lambda\in [0,1]$ and the robust algorithm $A$, and which we denote by \hybrid; we will also say that \hybrid is {\em based} on $A$. Let $a,b\in \mathbb{N}^+$ be such that $\lambda=a/(a+b)$. 
We require that the parameter $m$ in the statement of \pp is a sufficiently large constant, namely 
\begin{align}
m \geq 5\tau_k \max\{k,a+b\}/\epsilon,
\label{eq:mlower}
\end{align}
as it will become clear in the proof of Theorem~\ref{th:hybr-theory}.

Upon arrival of an item of size $x\in [1,k]$, \hybrid marks it as either an item to be served by \pp, or as an item to be served by $A$; we call such an item a {\em PP-item} or an {\em $A$-item}, in accordance to this action. Moreover, for every $x\in [1,k]$, \hybrid maintains two counters: \cnt, which is the number of items of size $x$ that have been served so far, and \ppcnt, which is the number of PP-items of size $x$ that have been served so far. 

We describe the actions of \hybrid. 
Suppose that an item of size $x$ arrives. If there is an empty placeholder of size $x$ in a non-empty bin, then the item is assigned to that bin (and to the corresponding placeholder), and declared PP-item. Otherwise, there are two possibilities: If \ppcnt$\leq \lambda \cdot$ \cnt, then it is served using \pp
and is declared PP-item. If \ppcnt$> \lambda \cdot$ \cnt, then it is served using $A$ and 
declared $A$-item. Thus, $\lambda$ is a measure of the ``involvement'' of the two algorithms in serving a given sequence. For extreme values of $\lambda$, \hybrid reduces to one of its two components: if $\lambda=1$, then \hybrid reduces to \pp, whereas if $\lambda=0$, it reduces to $A$.


Note that in \hybrid, $A$ and \pp maintain their own bin space, so when serving according to one of these algorithms, only the bins opened by the corresponding algorithm are considered. Thus, we can partition the bins used by \hybrid into {\em PP-bins} and {\em $A$-bins}.

\begin{example}\label{ex:3}
	Suppose that $k=10$, $m=20$, and that the predictions $\bm{f'}$ are as in 
       Example~\ref{ex:2}. 
       Figure~\ref{fig:hybrid} illustrates the packing of the 
       \hybrid algorithm that is based on \firstfit as algorithm $A$, with a parameter $\lambda = 0.5$. 
	
	\begin{figure}[htb!]
		\centering
		\scalebox{.8}{\figTwo}
		 \caption{The packing of \hybrid on the input sequence $\sigma = 2,3,1,4,10,2,9,4,6,9,2,6,5$. 
              The predictions and profile are described in Example~\ref{ex:2}. 
              The total cost incurred by the algorithm on $\sigma$ is equal to 9, with \pp and \firstfit contributing 6 and 3 bins, respectively. 
              }
              \label{fig:hybrid} 
	\end{figure}
	
\end{example}

Before presenting the analysis of \hybrid, we review two other approaches towards robustifying \pp based on a given (robust) algorithm $A$, and we show that these approaches do not perform as well as  \hybrid.  Recall that \hybrid will first place an item $x$ to a placeholder of size $x$ if such a placeholder is available. The first approach could be to skip this step; that is, consider an algorithm that serves a fraction $\lambda$ of items of size $x$ as PP-items, and the remaining $1-\lambda$ fraction as $A$-items. To exemplify the problem with this approach, suppose that $k$ is divisible by $5$, and consider a situation where the prediction specifies that half of the items have size $0.6k$, and the other half have size $0.4k$. 
Then, \pp will use two placeholders per profile bin, one of size $0.6k$ and another of size $0.4k$ in the optimal packing of the profile set. Suppose that the input sequence $\sigma$ of length $n$, and it consists only of items of size $0.4k$ and $0.6k$, in non-decreasing order of size. Moreover, the total frequency of items of size $0.4k$ is equal to $0.5+\eta/2$, whereas the total frequency of items of size $0.6k$ is equal to $0.5-\eta/2$, for some $\eta>0$ (hence $\eta$ is the prediction error). Then, the above algorithm opens $\lambda n/2$ bins for PP-items, each including an item of size $0.4k$ and a placeholder of size $0.6k$. Given that the initial check for placeholders is skipped, only $\lambda n (0.5-\eta/2)$ item of size $0.6k$ are placed in these placeholders. That is, in the final packing of the algorithm, 
there are $\lambda \eta n/2$ bins with a single item of size $0.4k$ and an empty placeholder of size $0.6k$. \hybrid, on the other hand, places an item of size $0.6k$ in each of these bins, and thus saves $\lambda \eta n/4$ bins. 

A second approach could be along the lines of~\cite{MahdianNS12}, which describe a general method for combining an \emph{optimistic} algorithm that trusts the prediction (in our context, \pp) and a \emph{pessimistic} algorithm that ignores the prediction (in our context, the online algorithm $A$). The optimistic and pessimistic algorithms optimize for situations where the prediction is perfect and adversarial, respectively. \cite{MahdianNS12} showed that their combined algorithm attains the best of both worlds performance for problems such as load balancing and facility location. Their algorithm serves an input item using the optimistic one, if the optimistic algorithm would have incurred a total cost that would be bounded by a constant factor of the corresponding cost that would be incurred by the pessimistic algorithm, up to that point in time; otherwise, it serves the item using the pessimistic algorithm. In our context, this approach would treat an item as PP-item if the total cost of \pp on the prefix of the input observed so far is within a constant factor $\alpha>0$ of the cost of $A$ on the same prefix. Unfortunately, this approach fails for the bin packing problem. Consider a worst-case example, where half of the items are predicted to be of size $0.4k$, and the remaining half are predicted to be of size $0.6k$. Suppose that this prediction is error-free and all items of size $0.4k$ appear before items of size $0.6k$ in the sequence. Then \pp reserves place-holders of size $0.6k$ in its bins, and its cost on any prefix of the input formed by items of size $0.4k$ is exactly twice as large as the cost of $A$ for on the same prefix: this is because $A$ must place two such items in the same bin, because of its pessimistic nature. We observe that if $\alpha \leq 2$, this combined algorithm treats all items of size $0.4k$ as A-items, and its final packing will be the same as the one of $A$. Similarly, if $\alpha>2$, the algorithm treats all items as PP-items and thus reduces to \pp. In other words, this approach fails to meaningfully combine the two algorithms.

%

%

\subsection{Analysis of \hybrid}
The following theorem establishes the performance of \hybrid.

\begin{theorem}\label{th:hybr-theory}
	For any $\epsilon \in (0,0.5]$ and $\lambda \in[0,1]$, \hybrid has competitive ratio 
	$(1+\epsilon) ((1+(2+5\epsilon)\eta k + \epsilon)\lambda + c_A (1-\lambda) ) $, 
	where $c_A$ is the competitive ratio of the algorithm $A$ on which \hybrid is based. 
\end{theorem}
\begin{proof}
	We define two partitions of the multiset of items in $\sigma$.
	The first partition is $S_{PP} \cup S_A$, where $S_{PP}$ and $S_A$ are the PP-items and A-items of \hybrid, respectively. The second partition is $S'_{PP} \cup S'_{A}$, where $S'_{PP}$ and $S'_A$ are defined such that for any $x \in [1,k]$ there are $\lfloor \lambda n_x \rfloor$ items of size $x$ in $S'_{PP}$ and 
	$n_x - \lfloor \lambda n_x \rfloor$ items of size $x$ in  $S'_{A}$. Given $\opt(\sigma)$, we will define a new packing $N$, such that every bin in $N$ contains only items 
	in $S'_{PP}$ or only items in $S'_{A}$. Let $N_{PP}$ and $N_A$ denote the set of bins in $N$ that include items in $S'_{PP}$ and in $S'_{A}$, respectively. Similarly, let $B_{PP}$ and $B_{A}$ denote the set of bins in the packing of \hybrid that contain only PP-items (PP-bins) and A-items (A-bins), respectively.
	We will prove the following bounds for $N$, $B_{PP}$ and $B_A$:
\begin{itemize}
	\item[(i)] $|N| \leq (1+\epsilon)|\opt(\sigma)|$; 
\item[(ii)]
$|B_{PP}| \leq (1+(2+5\epsilon)\eta k + \epsilon) |N_{PP}|$; \item[(iii)]
$|B_A| \leq c_A |N_{A}|$.
	
\end{itemize}
	%

To prove the above bounds, we first explain how to derive $N$ from a given optimal packing 
$\opt(\sigma)$. This is done in a way that 
$N$ contains the filled bins of $\opt(\sigma)$ and up to $(a+b)\tau_k$ additional filled bins so as to guarantee that, for each bin type in $\opt(\sigma)$, 
the number of bins of each given type in $N$ is divisible by $a+b$; recall that $a,b$ are the smallest integers such that $\lambda = \frac{a}{a+b}$.
Using~\eqref{eq:mlower}, 
we obtain that 
\[ 
|N| \leq |\opt(\sigma)| + (a+b-1)\tau_k < 
|\opt(\sigma)| ( 1+ \tau_kk/m),
\] 
Since the number of bins of each type in $N$ is divisible by $a+b$, we can partition  $N$ into $N_{PP}$ and $N_{A}$ so that $|N_{PP}| \leq a (1+\epsilon) |\opt(\sigma)|/(a+b) $ and $|N_A| \leq b(1+\epsilon)|\opt(\sigma)|/(a+b)$. That is,  $|N_{PP}| \leq \lambda (1+\epsilon) |\opt(\sigma)|$ and $|N_A| \leq (1-\lambda) (1+\epsilon) |\opt(\sigma)|$. Note that $N$ packs not only items in $\sigma$ but also additional items in the added bins. That implies that all items $S'_{PP}$ are packed in $N_{PP}$ and all items in $S'_A$ are packed in $N_A$, and hence (i) follows. 

To prove (ii) and (iii), we note that $S_{A} \subseteq S'_A$ which implies $S'_{PP} \subseteq S_{PP}$. This is because the algorithm declares an item of size $x$ as an A-item only if \ppcnt$> \lambda$ \cnt. 
Hence, at any given time during the execution of \hybrid, the number of A-items of size $x$ is no more than a fraction $(1-\lambda)$ of  \cnt. 

Next, we will show the property (ii). 
First, note that 
$|B_{PP}| = |PP(\sigma_{PP}, \bm{f'})|$, where $\sigma_{PP}$ is the subsequence of $\sigma$ formed by the $PP$-items, and $PP$ abbreviates the output of \pp.
Consider a sequence $\sigma'_{PP}$ obtained by removing, for every $x \in [1,k]$, the last $d_x$ items of size $x$ from $\sigma_{PP}$, where $d_x$ is the number of items of size $x$ in $S_{PP} \setminus S'_{PP}$.  We show next that $|PP(\sigma_{PP},\bm{f'})| = |PP(\sigma'_{PP}, \bm{f'})|$.
For any $x$, consider the last PP-item $L_x$ of size $x$ for which \hybrid opens a new bin. At the time $L_x$ is packed, \ppcnt$\leq \lambda \cdot $\cnt. 
Thus, by removing items of size $x$ that appear after $L_x$ in  $\sigma_{PP}$, the remaining items form a subsequence of $\sigma'_{PP}$, and the number of bins does not decrease. That implies that $|PP(\sigma_{PP}, \bm{f'})| = |PP(\sigma'_{PP}, \bm{f'})|$. 
From Theorem~\ref{th:consistent}, we obtain 
\begin{align*}
|B_{PP}|  &= |PP(\sigma_{PP},\bm{f'})| \\
&= |PP(\sigma'_{PP},\bm{f'})| \\
&\leq (1+(2+5\epsilon)\eta k+\epsilon)|\opt(\sigma'_{PP})| \\
&\leq (1+(2+5\epsilon)\eta k+\epsilon) |N_{PP}|.
\end{align*}

Last, to show (iii), we note that the number of bins that \hybrid opens for items in $S_{A}$ is at most 
$c_A |\opt(S_A)| \leq c_A |\opt(S'_A)| \leq c_A |N_A|$. This is because $S_A \subseteq S'_A$. \

	Using Properties (i)-(iii), we obtain 
\begin{align*}
|{\sc Hybrid}(\sigma,\bm{f'})| &= \ |B_{PP}| + |B_{A}| \\
&\leq  \ (1+(2+5\epsilon)\eta k + \epsilon) |N_{PP}| + c_A |N_{A}|  \\
&\leq   \ (1+(2+5\epsilon)\eta k + \epsilon)\lambda(1+\epsilon) |\opt(\sigma)| + c_A (1-\lambda)(1+\epsilon)|\opt(\sigma)| \\
&=  (1+\epsilon)((1+(2+5\epsilon)\eta k + \epsilon)\lambda + c_A (1-\lambda)) |\opt(\sigma)|,
\end{align*}
which concludes the proof.
\end{proof}

To obtain the best theoretical performance, we can choose $A$ as the algorithm of the best known competitive ratio, that is Advanced Harmonic algorithm~\cite{BaloghBDEL18}. However, as discussed in Section~\ref{sec:prelim}, such algorithms belong to a class that is tailored to worst-case competitive analysis, and do not tend to perform well in typical instances~\cite{kamali2015all}. For this reason, simple algorithms such as \firstfit and \bestfit are preferred in practice~\cite{10.5555/241938.241940}. We obtain the following corollary.

\begin{corollary}\label{coro:hybrid:best2}
	For any $\epsilon \in (0,0.5]$ and $\lambda \in[0,1]$, there is an algorithm with competitive ratio $(1+\epsilon)(1.5783 + \lambda ((2+5\epsilon)\eta k -0.5783 + \epsilon))$. Furthermore, \hybrid based on \firstfit has competitive ratio $(1+\epsilon)(1.7 + \lambda ((2+5\epsilon)\eta k -0.7 + \epsilon))$.
\end{corollary}

%

We can do even better if an {\em upper bound} estimation of the error is known to the algorithm. Such an upper bound, which we denote by $H$, may be available depending on the quality of historical data and the characteristics of typical sequences. Specifically, let \haware denote the algorithm which executes {\sc Hybrid}(1), if $H < (c_A-1-\epsilon)/(k (2+5\epsilon))$, and {\sc Hybrid}(0), otherwise. An equivalent statement is that \haware executes \pp if $H < (c_A-1-\epsilon)/(k (2+5\epsilon))$, and $A$, otherwise. The following corollary follows directly from Theorem~\ref{th:hybr-theory} with the
observation that as long as $\eta < (c_A-1-\epsilon)/(k (2+5\epsilon))$, \pp has a competitive ratio that is better than $c_A$.

\begin{corollary}\label{th:h-aware}
	For any $\epsilon \in (0,0.5]$, \haware using algorithm $A$ has competitive ratio $\min\{c_A, 1+ (2+5\epsilon)\eta k+\epsilon\} $, where $c_A$ is the competitive ratio of $A$. In particular, choosing \firstfit 
	as $A$, \haware has competitive ratio
	$\min\{1.7, 1+ (2+5\epsilon)\eta k + \epsilon\}$.
	\label{thm:h-aware}  
\end{corollary}

\section{Applications \& Extensions}\label{sec:extensions}

In this section, we discuss extensions and further applications of our algorithms.

\subsection{Virtual Machine Placement}\label{sect:vmsection}
An important application of online bin packing is Virtual Machine (VM) placement in large data centers. Here, 
each VM corresponds to an item whose size represents the resource requirement of the VM, and each bin corresponds to a physical machine (i.e., host) of a given capacity $k$. In the context of this application, the {\em consolidation ratio}~\cite{Mann15} is the number of VMs per host, in typical scenarios. Note that the consolidation ratio is typically much smaller than $k$. 
For example, VMware server virtualization achieves a consolidation ratio of up to 15:1~\cite{VMWareVM}, while Intel's virtualization infrastructure gives a consolidation ratio of up to 20:1~\cite{IntelVM}.

Let $r$ denote the consolidation ratio (but note that this quantity is an integer).
The following result shows that we can express the competitive ratio of \hybrid in Theorem~\ref{th:hybr-theory} so that the capacity $k$ is replaced by the consolidation ratio $r$. We can thus exploit the fact that typically $r$ is much smaller than $k$, and improve the 
theoretical analysis of our algorithms. 

%
\begin{theorem} \label{th:vm}
	Consider an instance of online bin packing with bins of capacity $k$, in which the item sizes are such that at most $r$ items can fit into a bin, for some $r
	\leq k$. Then, 
	for any constant $\epsilon\in(0,0.2]$, and predictions $\bm{f'}$ with error $\eta$, the following hold: \emph{(I)} \pp has competitive ratio at most $1+(2+5\epsilon)\eta r +\epsilon$; and  \emph{(II)} for any  $\lambda \in[0,1]$, \hybrid has competitive ratio 
	$(1+\epsilon) ((1+(2+5\epsilon)\eta r + \epsilon)\lambda + c_A (1-\lambda) ) $, 
	where $c_A$ is the competitive ratio of the algorithm $A$ on which \hybrid is based. 
\end{theorem}

\begin{proof}
	The proof of (I) is identical to that of Theorem~\ref{th:consistent}, except that in the fourth inequality that bounds $|PP(\sigma, {\bf f'})|$, we use the fact that $p_{\bm f}\geq \lceil m/r \rceil$ (instead of $p_{\bm f}\geq \lceil m/k \rceil$), given that at most $r$ items can fit into each bin. Moreover, in all subsequent inequalities in the proof, $k$ is replaced with $r$. The proof of (II) is identical to that of Theorem~\ref{th:hybr-theory}, except that part (ii) of Theorem~\ref{th:hybr-theory} 
	is replaced by $|B_{PP}| \leq (1+(2+5\epsilon)\eta r + \epsilon) |N_{PP}|$, which directly follows from the same arguments and by applying part (I) instead of Theorem~\ref{th:consistent}.
\end{proof}

Similarly, we can generalize Theorem~\ref{th:tradeoff} and obtain the following improved impossibility result. The proof is identical to that of Theorem~\ref{th:tradeoff}, with $k$ replaced by $r$. 

\begin{theorem}\label{th:lowerr}
	Consider an instance of online bin packing with bins of capacity $k$, in which the item sizes are such that at most $r$ items can fit into a bin, for some $r
	\leq k$. Then, for any constant $c<1$, and for any $\alpha\leq c/r$,
	any algorithm with frequency predictions that is $(1+\alpha)$-consistent has robustness at
	least $(1-c)r/2$. \label{th:tradeoffVM}
\end{theorem}

\subsection{Handling Items with Fractional Sizes}\label{sec:frac}

As stated in Section~\ref{sec:prelim}, we assume a discrete model in which items have integral sizes in $[1,k]$. While this is a natural model for many AI applications, our algorithms can also handle {\em fractional} item sizes in $[1,k]$, by treating them as ``special'' items, in the sense that they are not predicted to appear. \pp and \hybrid will then pack these fractional items separately from all integral ones, using \firstfit.
For the analysis in this setting, we need a measure of ``deviation'' of the input sequence $\sigma$ (that may contain fractional items) from a sequence of integral sizes. The first, and perhaps most natural, approach is to define this deviation as the $L_1$ distance between $\sigma$, and the sequence in which each item is rounded to the closest integer in $[0,k]$. However, we show that this definition can be overly restrictive.



\begin{theorem}
	Let $\lfloor x \rceil$ denote the integer closest to $x$, and define
	$d(\sigma) = \sum\nolimits_{x\in \sigma}|x - \lfloor x \rceil|$. Then no online algorithm in the fractional setting can have a competitive ratio better than 4/3, even if all frequency predictions are error-free (that is, $\eta = 0$), and even if $d(\sigma)=\epsilon$, for arbitrarily small $\epsilon>0$.
	\label{thm:dev1}
\end{theorem}

 \begin{proof}
	Let $\sigma = \sigma_1 \sigma_2$, where $\sigma_1$ consists of $n$ items of size $0.5-\epsilon/(2n)$, and $\sigma_2$ consists of $n$ items of size $0.5+\epsilon/(2n)$. For simplicity, we assume that $n$ and $k$ are even integers. Suppose also that $\bm{f'}$ is such that $f_{x,\lfloor \sigma \rceil}=1$, if $x=k/2$, and 0, otherwise (i.e., only items of size $k/2$ are predicted to appear in $\sigma$).
	From the definition of error, it also follows that $\eta=0$, and from the definition of the deviation $d$, we have that $d(\sigma)=\epsilon$.

	Let $A$ be any online algorithm, then from the definition of $\sigma$, we have that $A(\sigma_1,{\bm f'}) = cn$, for some $c \geq 1/2$.
	Given that $\opt(\sigma_1) =  n/2$, the competitive ratio of $A$ is at least $2c$. Out of the $cn$ bins of $A(\sigma_1, \bm{f'})$, $n-cn$ bins must have two items,  whereas the remaining $cn - (n-cn) = 2cn - n$ bins must have one item. Any of these remaining bins can each accommodate another item from $\sigma_2$. Therefore, out of the $n$ items in $\sigma_2$, $A$ can pack at most  $2cn-n$ such items in the $cn$ bins opened for $\sigma_1$, and it must place the remaining $n - (2cn-n) = 2n-2cn $ items in separate (new) bins. It follows that $A(\sigma,{\bm f'}) \geq cn + (2n-2cn) = 2n-cn$. Given that $\opt(\sigma) = n$, the competitive ratio of $A$ is therefore at least $2-c$. In summary, the competitive ratio of $A$ is $\max\{2c,2-c\}$, which is minimized at $4/3$ for $c=2/3$.
\end{proof}

In light of the above negative result, a different measure of ``deviation'' can be defined the ratio between the total size of fractional items in $\sigma$ over the total size of all items in $\sigma$. The following theorem shows that this measure can better capture the performance of the algorithm in the fractional setting.

\begin{theorem}
	Define $\hat{d}(\sigma)=\frac{\sum_{x \in \sigma,x \neq \lfloor x \rceil} x}{\sum_{x\in \sigma} x}$.
	Let $A$ be any algorithm with frequency predictions that has competitive ratio $c$ if all items have integral size. Then there is an algorithm $A'$ that has competitive ratio at most 
	$c + 2 \hat{d}(\sigma)$ for inputs with fractional sizes. 
	\label{thm:dev2}
\end{theorem}

 \begin{proof}
	Let $\sigma_I$ and $\sigma_F$ be the subsequences of $\sigma$ formed by integer and fractional items, respectively. We can write $A(\sigma) = A(\sigma_I) + FF(\sigma_F)$, where $FF(\sigma_F)$ denotes the number of bins opened by \firstfit when serving $\sigma_F$. For the number of bins opened for integer items, we have $A(\sigma_I) \leq A(\sigma) \leq c \cdot  \opt(\sigma)$. Let $S(\sigma)$ and $S(\sigma_F)$ denote the total size of items in $\sigma$ and $\sigma_F$, respectively, that is $S(\sigma) = \sum_{x\in \sigma} x$, and $S(\sigma_F) = \sum_{x \in \sigma,x \neq \lfloor x \rfloor} x$. 
	From definition, we have $\hat{d}(\sigma) = S(\sigma_F)/S(\sigma)$.
	Note that $FF(\sigma_F) \leq 2 S(\sigma_F)/k + 1$; this is because any pair of consecutive bins contains items of total size $k/2$ or larger. Therefore, 
       \[
       FF(\sigma_F) \leq 2 \hat{d}(\sigma) S(\sigma)/k + 1 \leq 2 \hat{d}(\sigma)\opt(\sigma) + 1.
       \] 
       In summary, we have $A(\sigma) \leq c \cdot \opt(\sigma) + 2\hat{d}(\sigma)\opt(\sigma) +1 $, therefore the (asymptotic) competitive ratio of $A$ is at most $c + 2\hat{d}(\sigma) $.
\end{proof}

\begin{example}
Suppose that the prediction specifies that half of the items are of size 4 and the remaining half are of size 6. Suppose also that the input $\sigma$ consists of $n/2$ items of size 4, $9n/20$ items of size 6, and $n/20$ items of size $6.1$.
Then, $\hat{d}(\sigma)=\frac{\frac{n}{20} \cdot 6.1}{\frac{n}{2} \cdot 4 + \frac{9n}{20} \cdot 6 + 
\frac{n}{20} \cdot 6.1} = \frac{6.1}{101.1} <0.061$.  Theorem~\ref{thm:dev2} shows that the worst-case, asymptotic competitive ratio of the algorithm cannot exceed $c+0.0122$ in the fractional setting. 
\end{example}

\subsection{A Sampling-based Algorithm for Online Bin Packing}
\label{subsec:sampling}

Our analysis of \pp, as stated in Theorem~\ref{th:consistent}, in conjunction with the PAC-learnability of frequency predictions, can help obtain a {\em sampling-based} algorithm with an efficient tradeoff between the number of sampled items and its attained competitive ratio. More precisely, consider the setting in which the online algorithm is allowed to observe $s$ items of the request sequence, and we would like to express its (asymptotic) competitive ratio as a function of $s$. Similar types of sampling-based competitive analysis have recently attracted attention in the context of other online problems such as ski rental and prophet inequalities~\cite{DBLP:conf/icml/DiakonikolasKTV21}, matching~\cite{DBLP:conf/soda/KaplanNR22}, and network optimization problems~\cite{neurips22:sampling}. 

Given any small constant $\epsilon>0$, define $\delta=1/\sqrt{2^{s\epsilon^2-k}}$. Let ON$^*$ denote the best online algorithm in the standard setting, which is currently the Advanced Harmonic algorithm~\cite{BaloghBDEL18} with competitive ratio 1.5783. We define {\sc Random-Mix} to be the algorithm that works as follows: With probability $\delta$, {\sc Random-Mix} executes ${\text ON}^*$, whereas, with probability $1-\delta$, it executes \pp. The analysis of this algorithm follows directly from Theorem~\ref{th:consistent} 
and Remark~\ref{remark:pac}. 

\begin{corollary}
For any constant $\epsilon>0$ and $k \in \mathbb N^+$, {\sc Random-Mix} with $s$ samples has expected competitive
ratio
$
1.5783\cdot \delta +(1-\delta)(1+(2+5\epsilon)\eta k+\epsilon),  
$
where $\delta=1/\sqrt{2^{s\epsilon^2-k}}$.
\label{cor:random}
\end{corollary}

Note that Corollary~\ref{cor:random} bounds the expected competitive ratio of a randomized algorithm which commits to its choice (that is, it executes either ${\text ON}^*$ or \pp, and this decision is made once the sample is revealed). In contrast, Corollary~\ref{coro:hybrid:best2} expresses the competitive ratio of a deterministic algorithm which judiciously switches between ${\text ON}^*$ and \pp throughout its execution, in order to achieve deterministic guarantees.


\section{Experimental Evaluation}
\label{sec:experiments}
In this section, we present an experimental evaluation of the performance of our algorithms\footnote{The code on which the experiments are based is available at \url{https://github.com/shahink84/BinPackingPredictions}.}. Specifically, in Section~\ref{subsec:benchmarks} we describe the benchmarks and the input generation model; in Section~\ref{exp:error}, we expand on the predictions and error measurement; and in Section~\ref{subsec:discussion}, we present and discuss the main experimental results. In addition, in Section~\ref{exp:profile} we report further experiments on the profile size, and in Section~\ref{sect:AvgExp} we provide further methodology for reporting the average performance of our algorithms over multiple runs. Last, in Section~\ref{sec:dynamic}, we study the performance of our algorithms in dynamic settings in which the input is generated from an evolving distribution. 

\subsection{Benchmarks and Input Generation}
\label{subsec:benchmarks}
Several benchmarks have been used in previous work on {\em offline} bin packing; we refer to the discussion by~\cite{CastineirasCO12} for a list of related work. Many of these previous benchmarks typically rely on either uniform or normal distributions. There are two important issues to take into account. First, such simple distributions are often unrealistic and do not capture typical applications of bin packing such as resource allocation, as observed in~\cite{Gent98}. Second, in what concerns online algorithms, simple algorithms such as \firstfit and \bestfit are very close to optimal for input sequences generated from uniform distributions~\cite{10.5555/241938.241940} and very often outperform, in practice, many online algorithms of better competitive ratio~\cite{kamali2015all}. 

We evaluate our algorithms on two types of benchmarks. The first type is based on the {\em Weibull}
distribution, which was first proposed in~\cite{CastineirasCO12} as a model of several real-world applications of bin packing, e.g., the 2012 ROADEF/EURO Challenge on a data center problem provided by Google and several examination timetabling problems. 
The Weibull distribution is specified by two parameters: the \emph{shape} parameter $sh$ and the \emph{scale} parameter $sc$ (with $sh,sc >0$).  The shape parameter defines the spread of item sizes: lower values indicate greater skew towards smaller items. The scale parameter 
represents the statistical dispersion of the distribution. In our experiments, we chose $sh \in [1.0,4.0]$. This is because values outside this range result in trivial sequences with items that are generally too small (hence easy to pack) or too large (for which any online algorithm tends to open a new bin). The scale parameter is not critical, since we scale items to the bin capacity, as we will discuss later; we thus set $sc=1000$, in accordance with~\cite{CastineirasCO12}. 

%

The second type of benchmarks is generated from the BPPLIB library~\cite{benchmarks}, a collection of bin packing benchmarks used in various works on (offline) algorithms for bin packing. 
In particular, we report results on the benchmarks ``GI"~\cite{GschwindI16},  ``Schwerin"~\cite{schwerin1997bin}, ``Randomly\_Generated"~\cite{delorme2014bin}, ``Schoenfield\_Hard28"~\cite{schoenfield2002fast} and ``W{\"a}scher"~\cite{wascher1996heuristics}.




We fix the size of the sequence to $n=10^6$.
We set the bin capacity to $k=100$, and we also scale down each item to the closest integer in $[1,k]$. 
This choice is relevant for applications such as Virtual Machine placement (Section~\ref{sect:vmsection}), as explained in Section~\ref{sect:vmsection}. 
We generate two classes of input sequences.
%
%
For Weibull benchmarks, the input sequence consists of items generated independently and uniformly at random, and the shape parameter is set to $sh=3.0$. For BPPLIB benchmarks, we first select a file of
the benchmark uniformly at random, then generate input items from the chosen file, again uniformly at random.

\subsection{Compared Algorithms, Predictions and Error}
\label{exp:error}

We evaluate \hybrid for $\lambda \in \{0, 0.25, 0.5, 0.75, 1 \}$, based on \firstfit. This means that {\sc Hybrid}(0) is identical to \firstfit, whereas {\sc Hybrid}(1) is identical to \pp. We fix the size of the profile set to $m=5000$. To simplify the 
implementation of \pp, we use the algorithm \firstfitdec~\cite{10.5555/241938.241940} to compute the profile packing, instead of an optimal algorithm. Specifically, \firstfitdec first sorts items in the non-increasing order of their sizes and then packs the sorted sequence using \firstfit.  Using \firstfitdec helps reduce the time complexity, and \remove{, in particular with regards to \adaptive, which must compute a new profile packing
	every time it updates the frequency prediction. }the results only improve by using an optimal algorithm for profile packing, instead. 

We generate the frequency predictions to \hybrid as follows: For a parameter $b \in \mathbb{N}^+$, we define the predictions $\bm{f'}$ as $\bm{f_{\sigma[1,b]}}$. In words, we use a prefix of size $b$ of the input $\sigma$ so as to estimate the frequencies of item sizes in $\sigma$.  In our experiments, we consider 100 different prefix sizes. More precisely, we consider all $b$ of the form 
$b=\lfloor 100 \cdot 1.05^i \rfloor$, with $i \in [25,125]$.
We define the prediction error $\eta$ as the $L_1$ distance between the predicted and the actual frequencies. Note that for a given input sequence, $\eta$ is a function of the prefix size $b$. Since we consider 100 distinct values for $b$, for each sequence we consider up to 100 possible error values. As expected from Remark~\ref{remark:pac}, the prediction error decreases with $b$.

As explained earlier, \firstfit and \bestfit perform very well in practice, and we use them as benchmarks for comparing our algorithms. As often in offline bin packing, we also report the {\em L2 lower bound}~\cite{MartelloT90,fukunaga2007bin} as a lower-bound estimation of the optimal {\em offline} bin packing solution. That is, no algorithm, online or offline, can perform better than this lower bound.


\begin{figure*}[!htb] 
	\ \vspace*{.5cm} \\
	\centering
	\begin{subfigure}[b]{0.495\textwidth}
		\includegraphics[width=1\linewidth,trim = 0mm 6mm 0mm 2cm, clip]{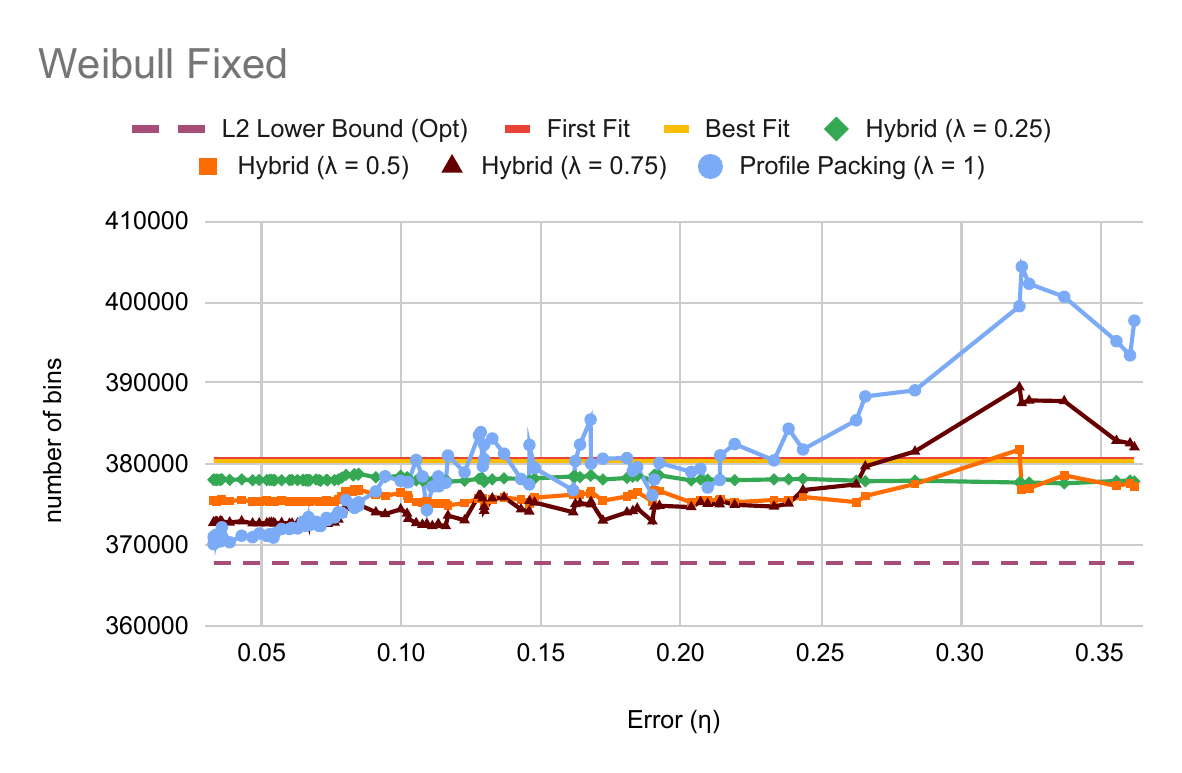}
		\caption{\scriptsize Weibull distribution.}
		\label{fig:fixed.WB}
	\end{subfigure}   
	\hfill 
	\begin{subfigure}[b]{0.495\textwidth}
		\includegraphics[width=1\linewidth, trim = 0mm 6mm 0mm 2cm, clip]{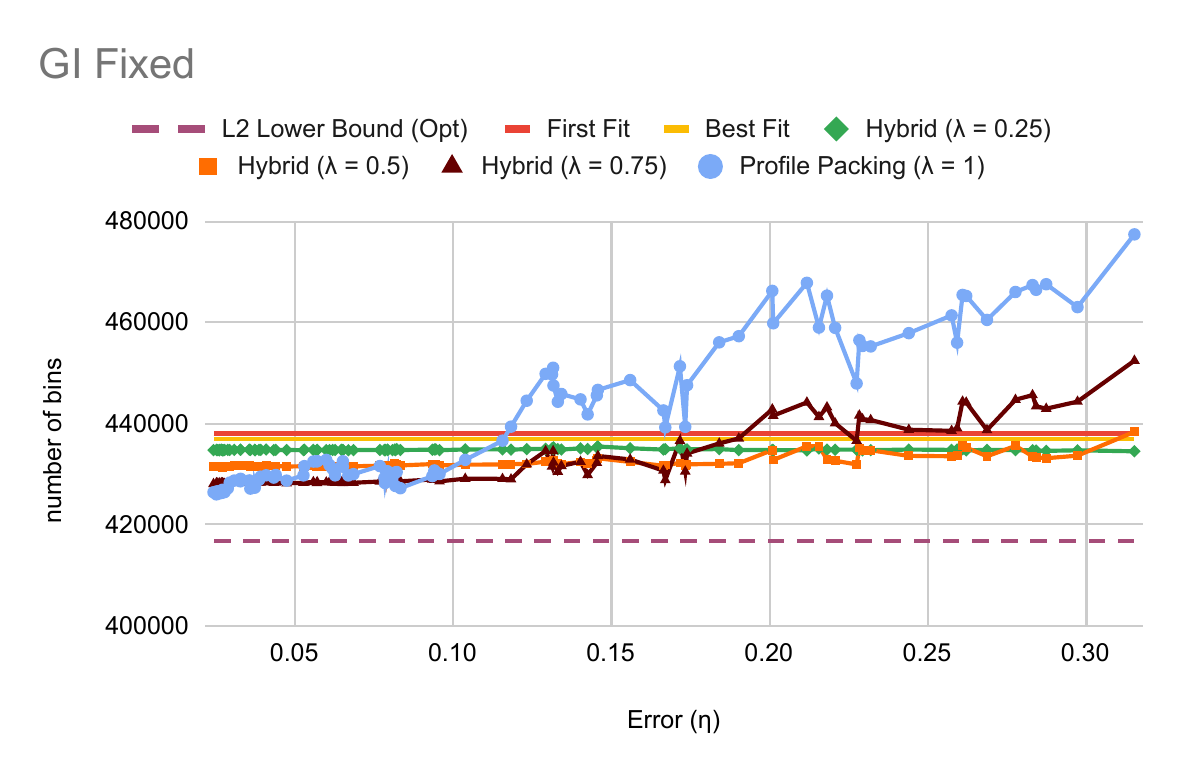}
		\caption{\scriptsize GI benchmark from BPPLIB.}
		\label{fig:fixed.GI}
	\end{subfigure}
	\ \vspace*{.3cm} \\
	\begin{subfigure}[b]{.495\textwidth}
		\includegraphics[width=1\linewidth,trim = 0mm 0mm 0mm 2cm, clip]{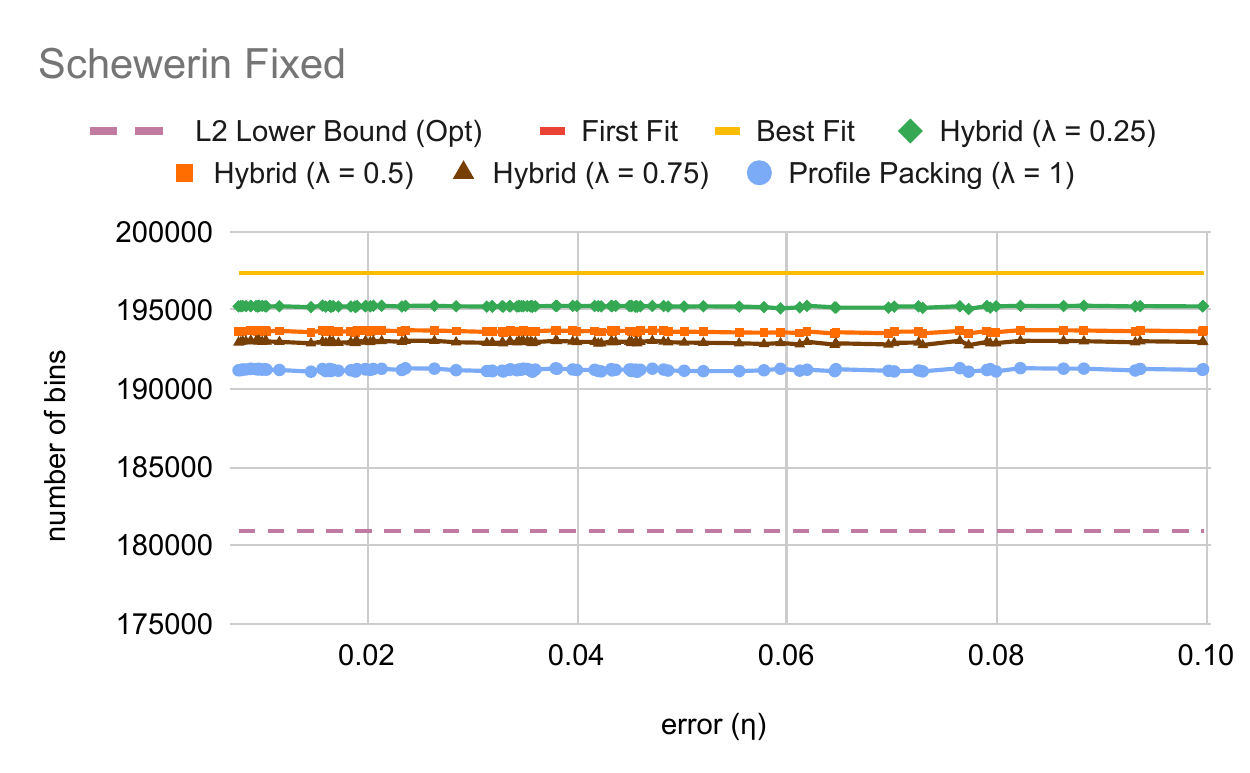} 
		\caption{\scriptsize Shwerin benchmark from BPPLIB.}
		\label{fig:fixedSchwerin}
	\end{subfigure}
	\begin{subfigure}[b]{.495\textwidth}
		\includegraphics[width=1\linewidth, trim = 0mm 0mm 0mm 2cm, clip]{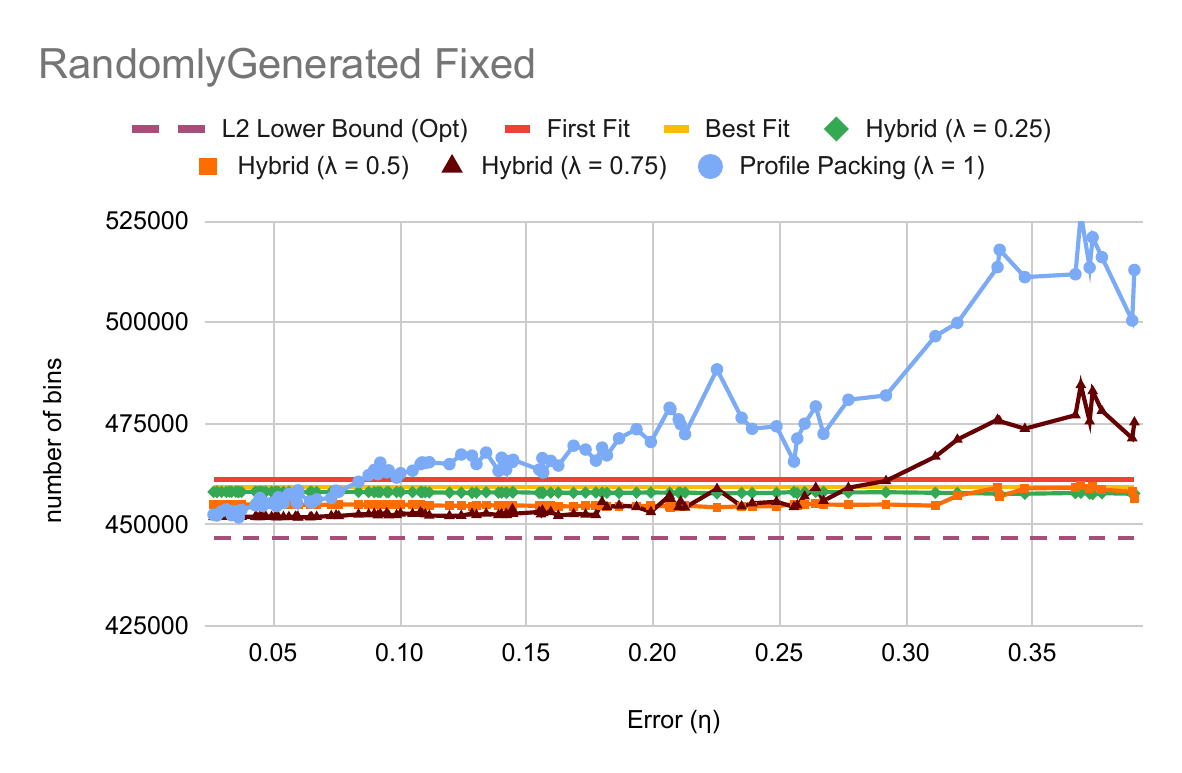}
		\caption{\scriptsize Randomly\_Generated benchmark from BPPLIB.}
		\label{fig:fixedRand}
	\end{subfigure}
	\ \vspace*{.3cm} \\
	\begin{subfigure}[b]{.495\textwidth}
		\includegraphics[width=1\linewidth, trim = 0mm 0mm 0mm 2cm,clip]{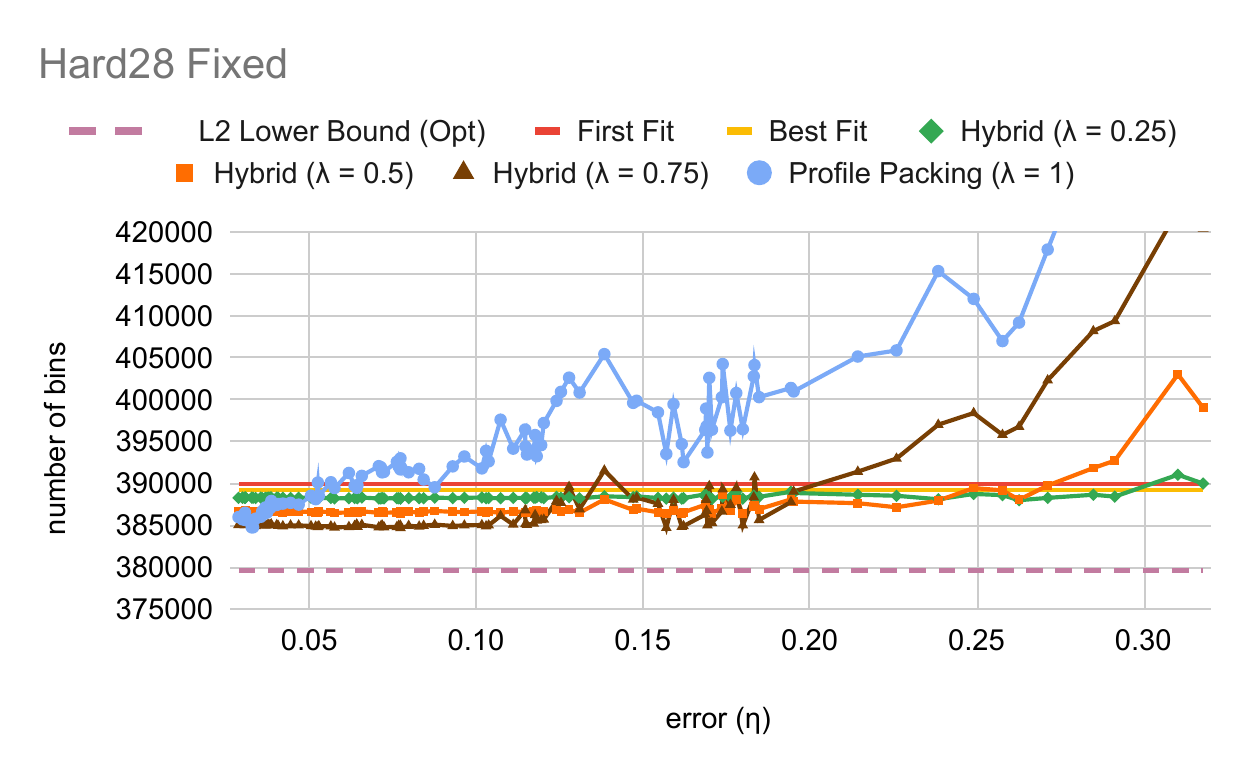}
	\caption{\scriptsize Schoenfield\_Hard28 benchmark from BPPLIB.}
	\label{fig:fixedhard}
\end{subfigure}
\begin{subfigure}[b]{.495\textwidth}
	\includegraphics[width=1\linewidth, trim = 0mm 0mm 1cm 2cm, clip]{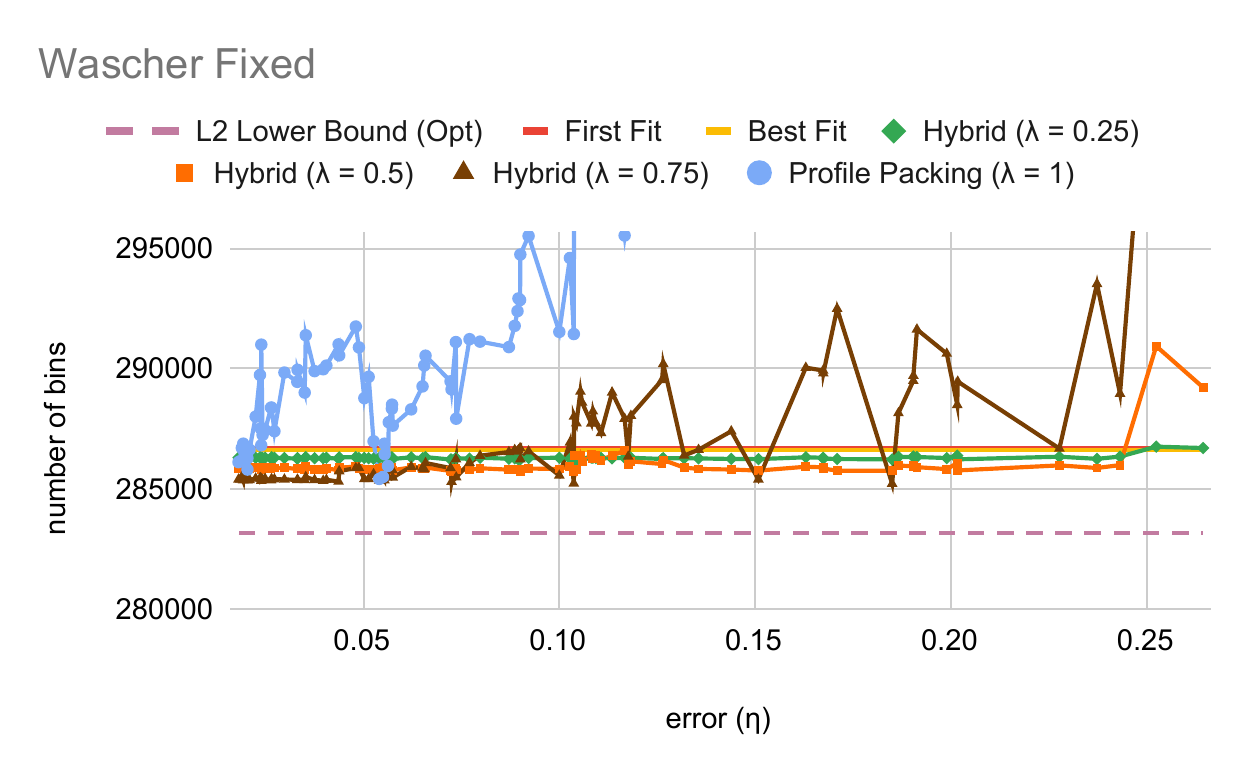}
\caption{\scriptsize W{\"a}scher benchmark from BPPLIB.}
\label{fig:fixedwascher}
\end{subfigure}
\caption{Number of opened bins for sequences from a given distribution. For the purpose of visualization, some of the plots are truncated, e.g., the plot of \pp in (c) and (d).}
\label{fig:fixedOrg}
\end{figure*}

\subsection{Results and Discussion}

\label{subsec:discussion}

Figure~\ref{fig:fixedOrg} depicts the cost of the algorithms for a typical sequence, as a function of the prediction error. The chosen files are ``csBA125\_9" (for ``GI"), ``Schwerin2\_BPP32"  (for ``Shwerin"), ``BPP\_750\_50\_0.1\_0.8\_2" (for ``Randomly\_Generated"), ``Hard28\_BPP832" (for ``Schoenfield\_Hard28"), and ``Waescher\_TEST0082" (for “W{\"a}scher”). 
Here, we consider a single sequence, as opposed to averaging over multiple sequences, because each input sequence is associated with its own prediction error, for any given prefix size (and naively averaging over both the cost and the error may produce misleading results). We can use a single sequence because the input size is considerable ($n=10^6$), and the distribution is fixed. Nevertheless, in Section~\ref{sect:AvgExp} we explain how to properly average over multiple sequences, and we report similar plots and conclusions.  

For all benchmarks, we observe that \pp $(\lambda=1)$ degrades quickly as the error increases, even though it has very good performance for small values of error. As $\lambda$ decreases, we observe that \hybrid becomes less sensitive to error, which confirms the statement of Corollary~\ref{coro:hybrid:best2}. 

Specifically, we observe that for the Weibull benchmarks, \hybrid dominates both \firstfit and \bestfit for all $\lambda \in\{0.25, 0.5, 0.75\}$ and for all $\eta<0.27$, approximately. For the GI benchmarks, \hybrid dominates \firstfit and \bestfit for $\lambda \in\{0.25, 0.5\}$, and for practically all values of error.
In the ``Shwerin" benchmark, all items have sizes in the range $[15,20]$. As such, very good predictions can be obtained by observing a tiny part of the input sequence, i.e., for small values of the prefix size $b$. 
In particular, the smallest value of $b$, namely $b=391$ results in $\eta < 0.099$.
As illustrated in Figure~\ref{fig:fixedSchwerin}, the smaller the parameter $\lambda$, the better the performance of \hybrid; in particular, \pp performs the best, which suggests that for inputs from a small set of item sizes, it is beneficial to choose a small value of $\lambda$. This can be explained by the fact that the prediction error is relatively smaller for these types of inputs.
This finding can be useful in the context of applications such as VM placement: this is because there is only a small number of different VMs that can be assigned to any given physical machine, as we discussed in  Section~\ref{sect:vmsection}.
For the remaining benchmarks, namely ``Randomly\_Generated", ``Schoenfield\_Hard28", and “W{\"a}scher”, the relative performance of the algorithms is similar to that for the GI benchmark, with the difference that the divergence of the algorithms becomes observable at different values of the prediction error.


The results demonstrate that frequency-based predictions indeed lead to performance gains. Even for very large prediction error (i.e., a prefix size as small as $b=338$) \hybrid with $\lambda \leq 0.5$ outperforms both \firstfit and \bestfit, therefore the performance improvement comes by only observing a tiny portion of the input sequence.


\subsection{Experiments on the Profile Size}
\label{exp:profile}
In previous experiments, we assumed that the profile size is $m=5000$. In this section, we report experiments on other values of $m$. More precisely, we evaluated the performance on two random sequences of length $n=10^6$ in which the item sizes are generated using Weibull distribution (with $sh=3$) and the GI-benchmark, respectively, as detailed in Section 6.2. As before, we choose $k=100$. Predictions are generated based on a prefix of length $b=1000$ of the input; this resulted in error values of 
$\eta = 0.1922$ and $\eta = 0.2045$ for the Weibull and GI-instances, respectively. 
We run \hybrid ($\lambda \in \{0.25, 0.5, 0.75, 1\}$) for $100$ different values of $m$, equidistant in the interval $[100,100100]$.

\begin{figure*}[t!]
\ \vspace*{.75cm} \\
\centering
\begin{subfigure}[b]{.495\textwidth}
\includegraphics[width=\linewidth,trim = 0mm 5mm 0mm 2cm, clip]{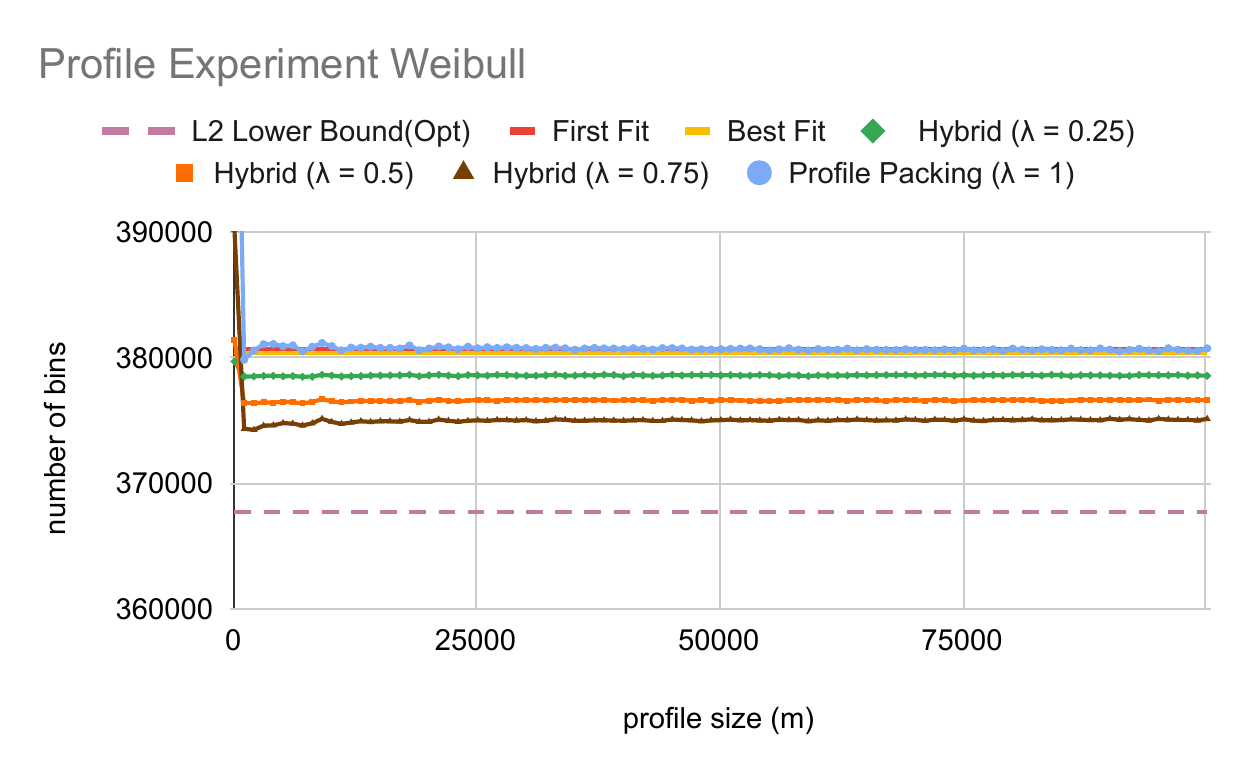} 
\caption{\scriptsize Weibull benchmark.}
\label{fig:dyn.Weibull}
\end{subfigure}
\hfill 
\begin{subfigure}[b]{.495\textwidth}
\includegraphics[width=\linewidth, trim = 0mm 5mm 0mm 2cm, clip]{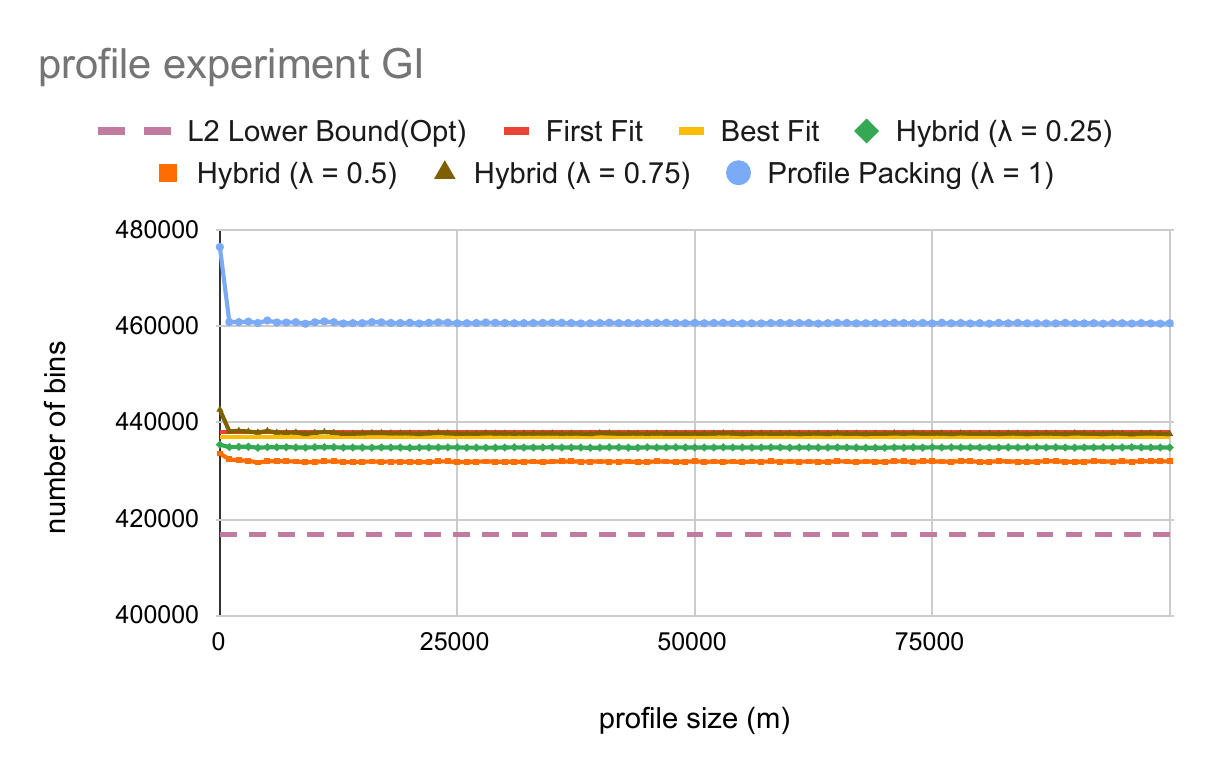}
\caption{\scriptsize GI benchmark from BPPLIB.}
\label{fig:ProfileFix}
\end{subfigure}
\caption{Number of opened bins as a function of the profile size.}
\label{fig:profile}
\end{figure*}

Figure~\ref{fig:profile} depicts the number of bins opened by the algorithms. The experiments show that the parameter $m$ has little impact on the performance of \hybrid, that is, as long as $m$ is sufficiently large (e.g., when $m \geq 1000$), the performance of \hybrid is consistent and independent of the choice of $m$.

\subsection{Experiments on the Average Cost and Rounded Error}
\label{sect:AvgExp}

In the experiments that we discussed in Section~\ref{subsec:discussion}, we reported the performance of the algorithm on a typical sequence. More precisely, we considered a single randomly generated sequence, as opposed to averaging the cost of the algorithm over multiple input sequences, because each input sequence is associated with its own prediction error, for any given size of the prefix (and averaging naively over both the cost and the error, simultaneously, may produce misleading results). We argued that this should not be an issue, because the input sequence is of considerable size ($n=10^6$). 

In this section, we present further experimental results based on averaging over both the cost and the error which give further justification for this choice. 
Our setting here is as follows: Given a fixed distribution (either Weibull with
$sh=3$, or a file from the GI Benchmark), we generate 20 random sequences of length $10^6$.
For each sequence, we compute \firstfit, \bestfit, and the $L2$ lower bound. The average costs of these algorithms, over the 20 sequences, serve as the benchmark costs for comparison.

For \hybrid, and every $\lambda \in [0.25, 0.5, 0.75,1]$, we generate predictions for $100$ values of 
the prefix size $b$ (where recall that $b$ is of the form $b = 100 \cdot 1.05^i$, with $i \in[25,125]$). Consider a sequence $\sigma$. For each of the above predictions for $\sigma$, we compute the prediction error as well as the cost of \hybrid on $\sigma$ with the corresponding prediction and store a pair of the form ({\sc error}, {\sc cost}), where {\sc error} is the error with a two-digit decimal precision, and the cost is the cost of the algorithm. For example, if {\sc error} = 0.2341 and {\sc cost} = 143000, we store the pair (0.23, 143000). This means that for a fixed sequence, we store up to 100 such pairs (assuming $\eta < 1$). Last, we evaluate the average of pairs with the same rounded error over the 20 sequences. For example, if for $\sigma_1$ 
we have obtained the pair (0.23,100000), for $\sigma_2$ the pair (0.23, 150000), and for $\sigma_3$ the 
pair (0.23, 350000), then we take the average as the pair (0.23, 200000). 

Figure~\ref{fig:average} depicts the plots obtained by this method, for both the Weibull and the GI benchmarks. We observe that \hybrid exhibits similar performance tradeoffs as the plots for a single sequence (Figure~\ref{fig:fixedOrg}), but the differences are less pronounced due to averaging.
\begin{figure*}[t]
\centering		\ \\ \vspace{.8cm}
\begin{subfigure}[b]{.495\textwidth}
\includegraphics[width=\linewidth,trim = 0mm 5mm 0mm 2cm, clip]{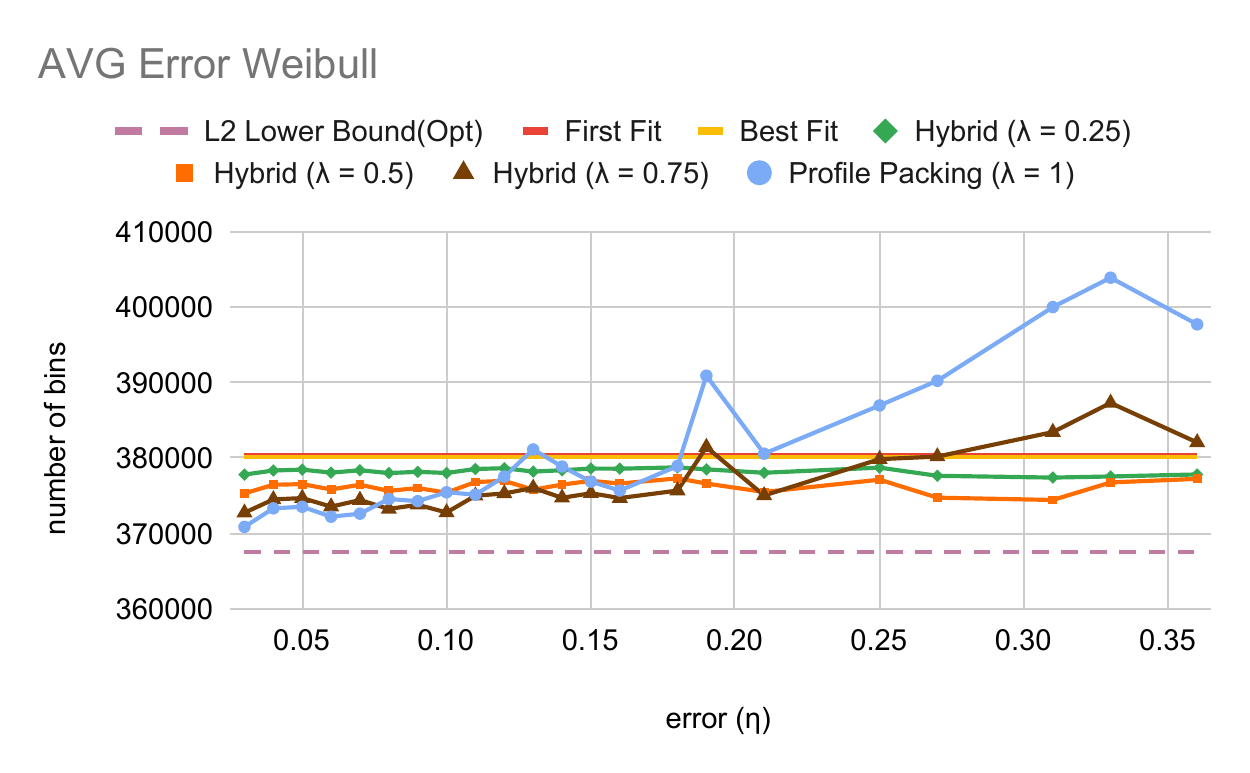} 
\caption{\scriptsize Weibull benchmark.}
\label{fig:average.weibull}
\end{subfigure}
\hfill 
\begin{subfigure}[b]{.495\textwidth}
\includegraphics[width=\linewidth, trim = 0mm 5mm 0mm 2cm, clip]{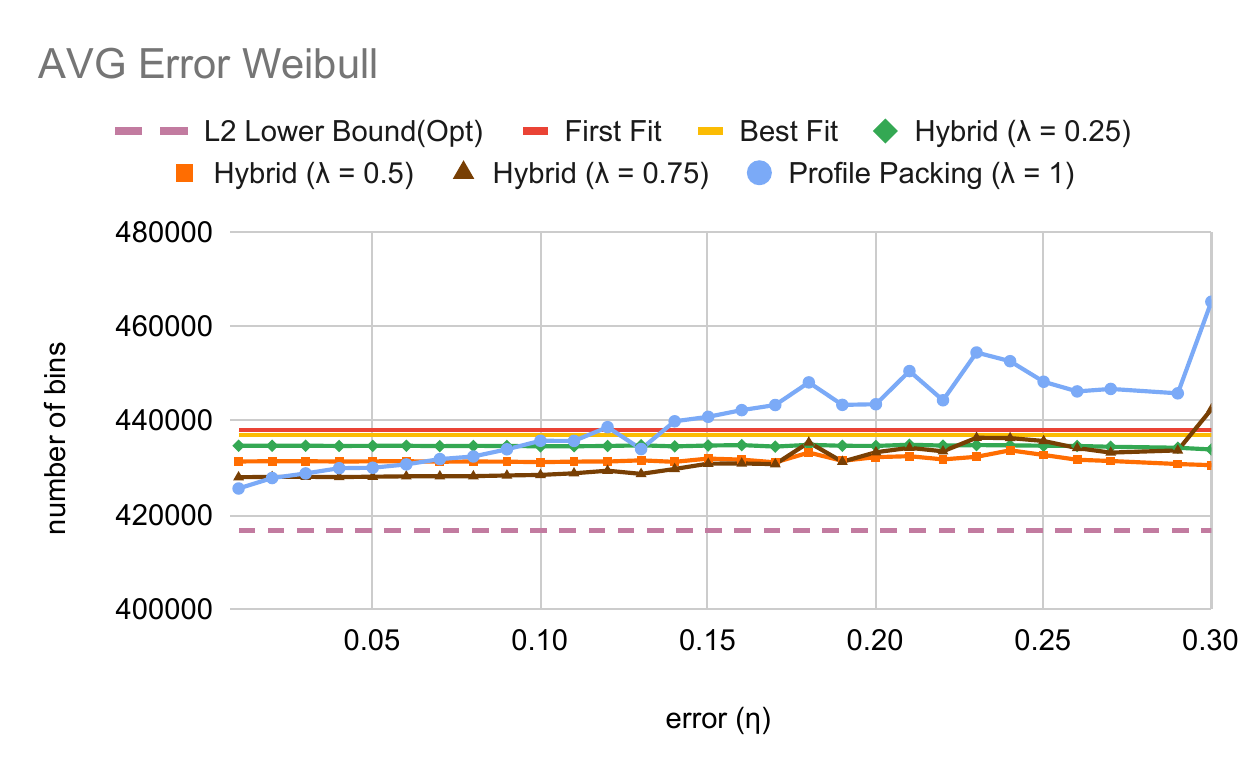}
\caption{\scriptsize GI Benchmark from BPPLIB.}
\label{fig:average.GI}
\end{subfigure}
\caption{Average number of bins vs. average error over twenty sequences. \label{fig:average} }
\end{figure*}

\subsection{Evolving Distributions}
\label{sec:dynamic}

In this section, we address the situation in which the input is not drawn according to a fixed distribution but instead is generated from distributions that change with time, e.g., when dealing with evolving data streams. This is a complex setting that has not been studied in any previous work on online bin packing, with or without predictions. 

We define a heuristic called \adaptive, in which predictions are updated dynamically using a {\em sliding window} approach; see e.g.~\cite{gomes2017survey}. 
\adaptive uses a parameter $w \in \mathbb{N^+}$ as follows. 
In the initial phase, \adaptive serves $\sigma[1,w]$ using \firstfit; moreover, at the end of this phase,
it computes $\bm{f_{\sigma[1,w]}}$, namely the frequency vector of all sizes in $\sigma[1,w]$. From this point onwards, the algorithm will serve items using \pp 
with predictions $\bm{f'}$ which are initialized to $\bm{f_{\sigma[1,w]}}$. Specifically, every time \adaptive encounters item $\sigma[i w]$, for $i \in \mathbb{N}^+$, it updates $\bm{f'}$ to
$\bm{f_{\sigma[(i-1)w+1,i w]}}$.

For the analysis, we use the benchmarks described in Section~\ref{subsec:benchmarks}.
The distribution of the input sequence changes every 50000 items. Namely, the input sequence is the concatenation of $n/50000$ subsequences. For Weibull benchmarks, each subsequence is a Weibull distribution, whose shape parameter is chosen uniformly at random from $[1.0,4.0]$. For BPPLIB benchmarks, each subsequence is generated by choosing a file uniformly at random, then generating $50000$ items uniformly at random from that specific file.

We evaluate \adaptive for 100 values of the sliding window $w$, equidistant in the range $[100,100000]$. This is a crucial parameter: if $w$ is too small, we do not obtain sufficient information on the frequencies, whereas if $w$ is too big, the predictions become ``stale''.

\begin{figure*}[!h]
\ \vspace*{.5cm} \\
\centering
\begin{subfigure}[b]{0.49\textwidth}   \includegraphics[width=1\linewidth,trim = 0mm 8mm 0mm 2.2cm, clip]{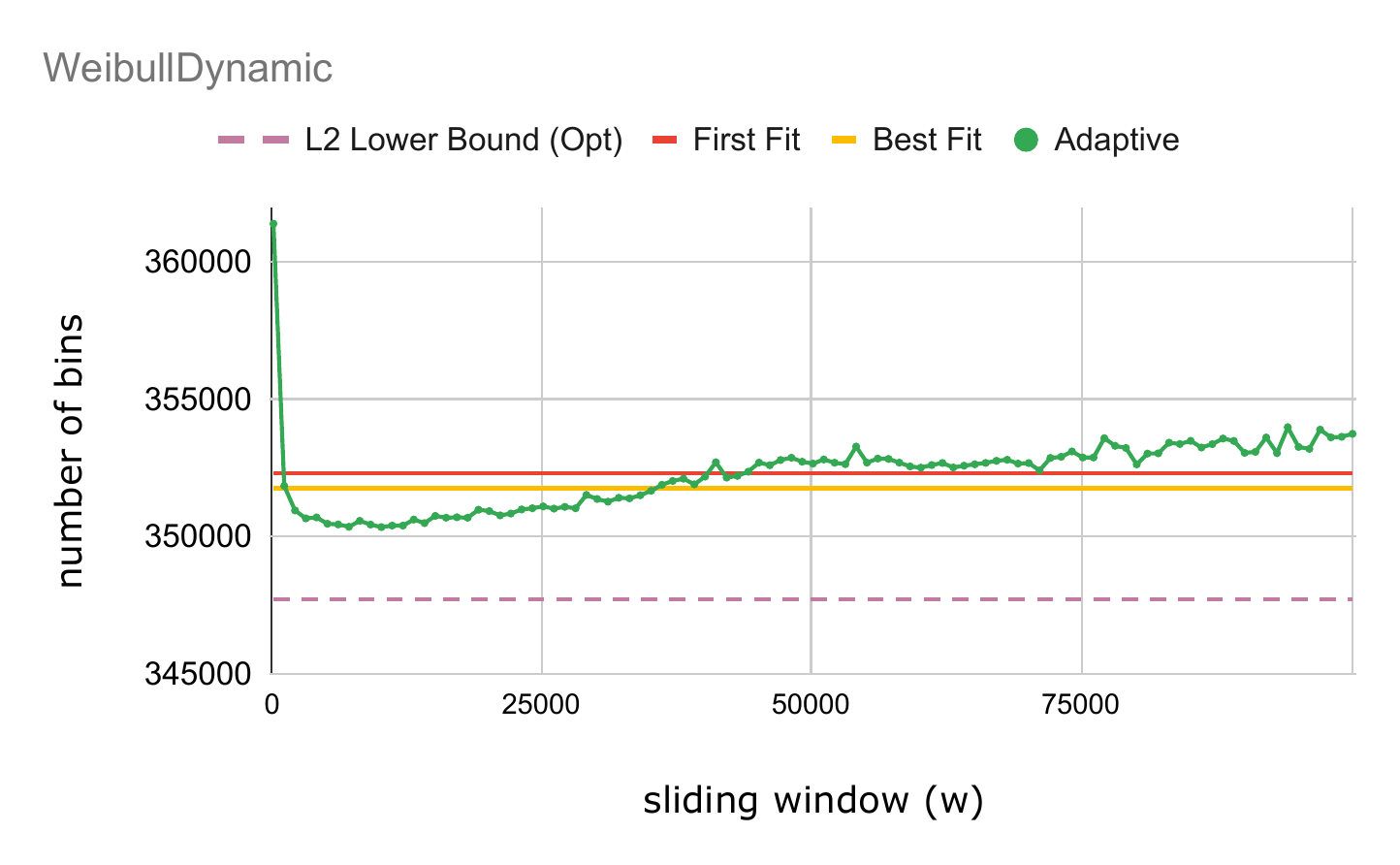}
\caption{\scriptsize  Weibull distribution.}
\label{fig:evolving.WB}
\end{subfigure}
\hfill 
\begin{subfigure}[b]{0.49\textwidth}    \includegraphics[width=1\linewidth, trim = 0mm 8mm 0mm 2cm, clip]{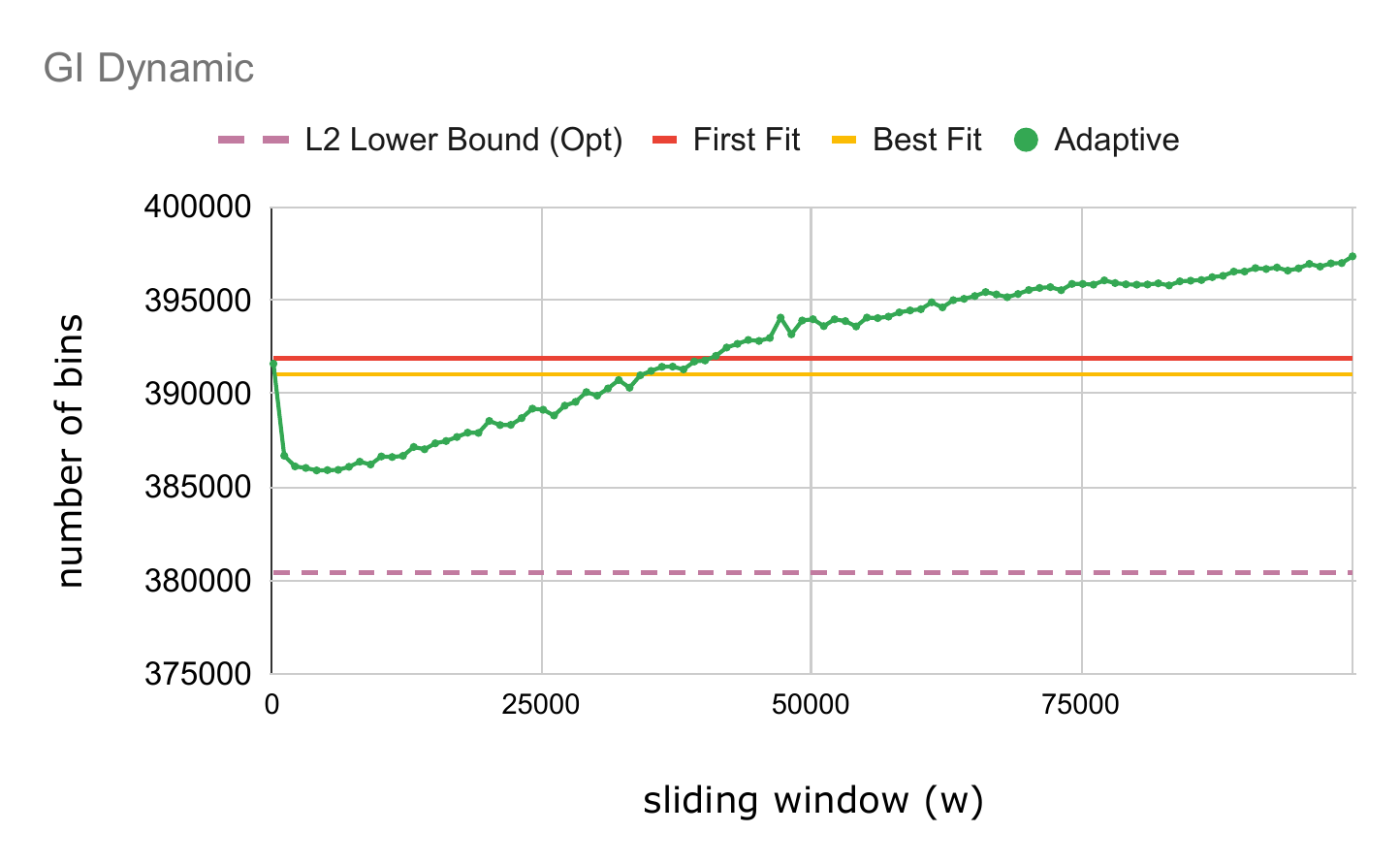}
\caption{\scriptsize GI benchmark from BPPLIB.}
\label{fig:evolving.GI}
\end{subfigure}\ \vspace*{.4cm} \\
\begin{subfigure}[b]{.49\textwidth}
\includegraphics[width=\linewidth,trim = 0mm 5mm 0mm 2cm, clip]{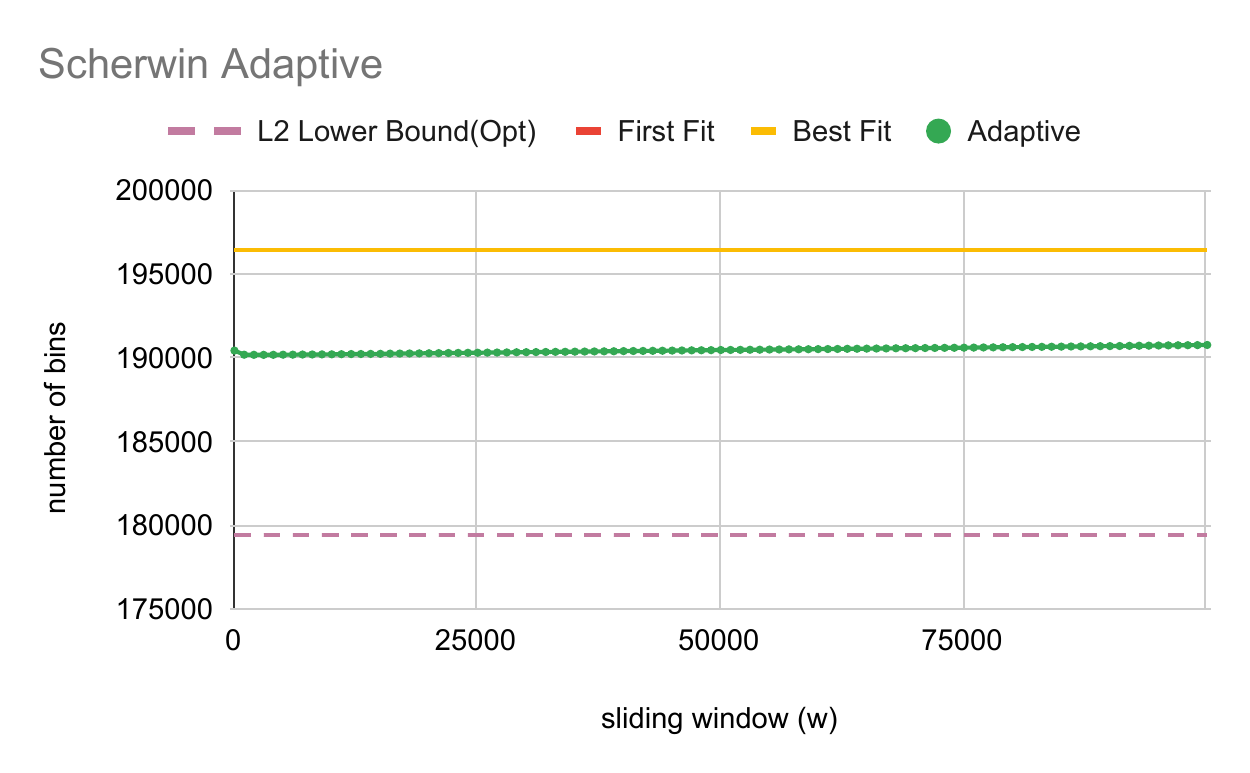} 
\caption{\scriptsize Shwerin benchmark from BPPLIB.}
\label{fig:dyn.Schwerin}
\end{subfigure}
\hfill 
\begin{subfigure}[b]{.49\textwidth}
\includegraphics[width=\linewidth, trim = 0mm 5mm 0mm 2cm, clip]{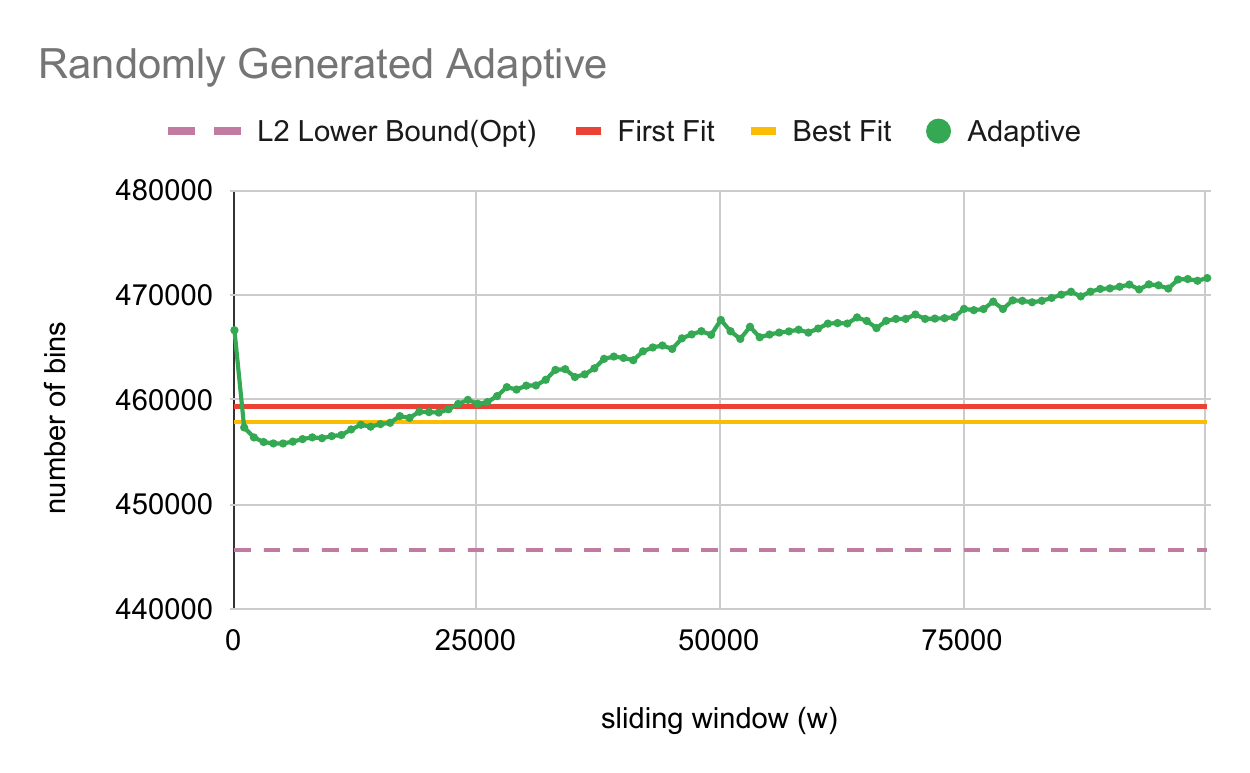}
\caption{\scriptsize Randomly\_Generated benchmark from BPPLIB.}
\label{fig:dyn.Rand}
\end{subfigure}
\ \vspace*{.4cm} \\
\begin{subfigure}[b]{.49\textwidth}
\includegraphics[width=\linewidth, trim = 0mm 5mm 0mm 2cm, clip]{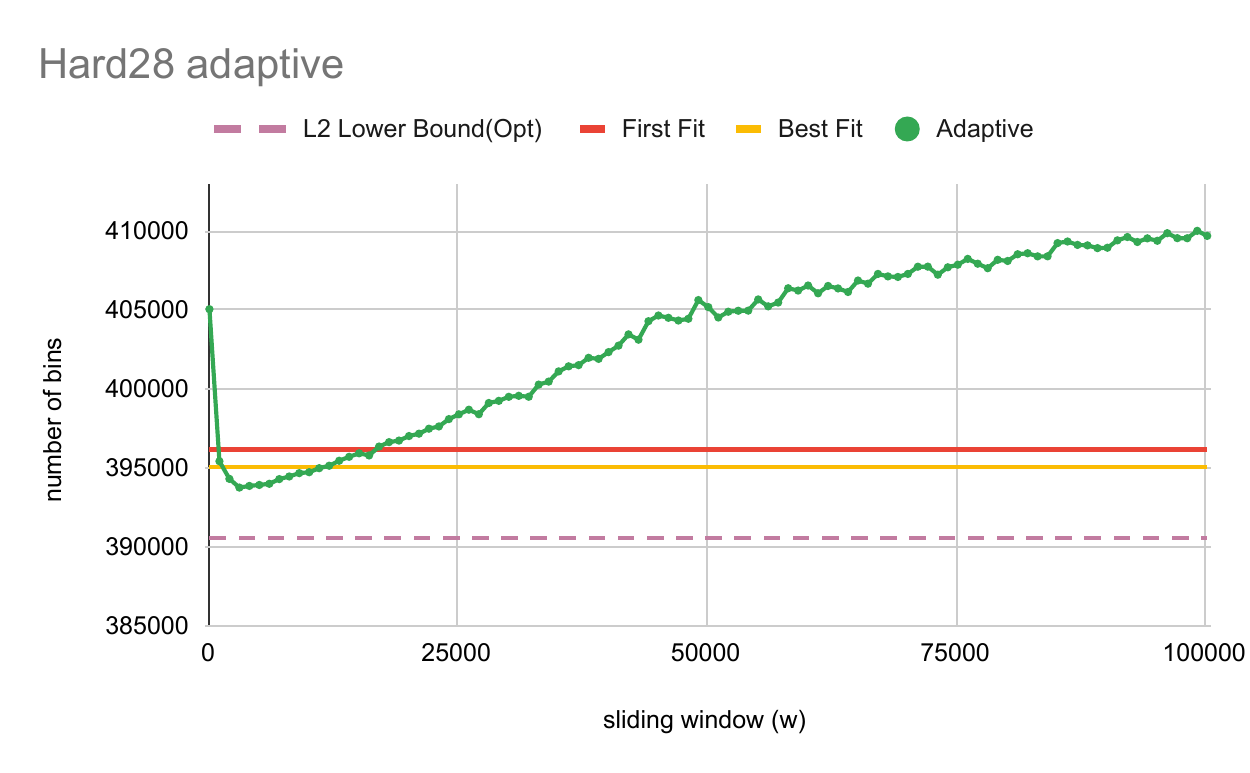}
\caption{\scriptsize Hard28 benchmark from BPPLIB.}
\label{fig:dyn.hard}
\end{subfigure}
\begin{subfigure}[b]{.49\textwidth}
\includegraphics[width=\linewidth, trim = 0mm 5mm 0mm 2cm, clip]{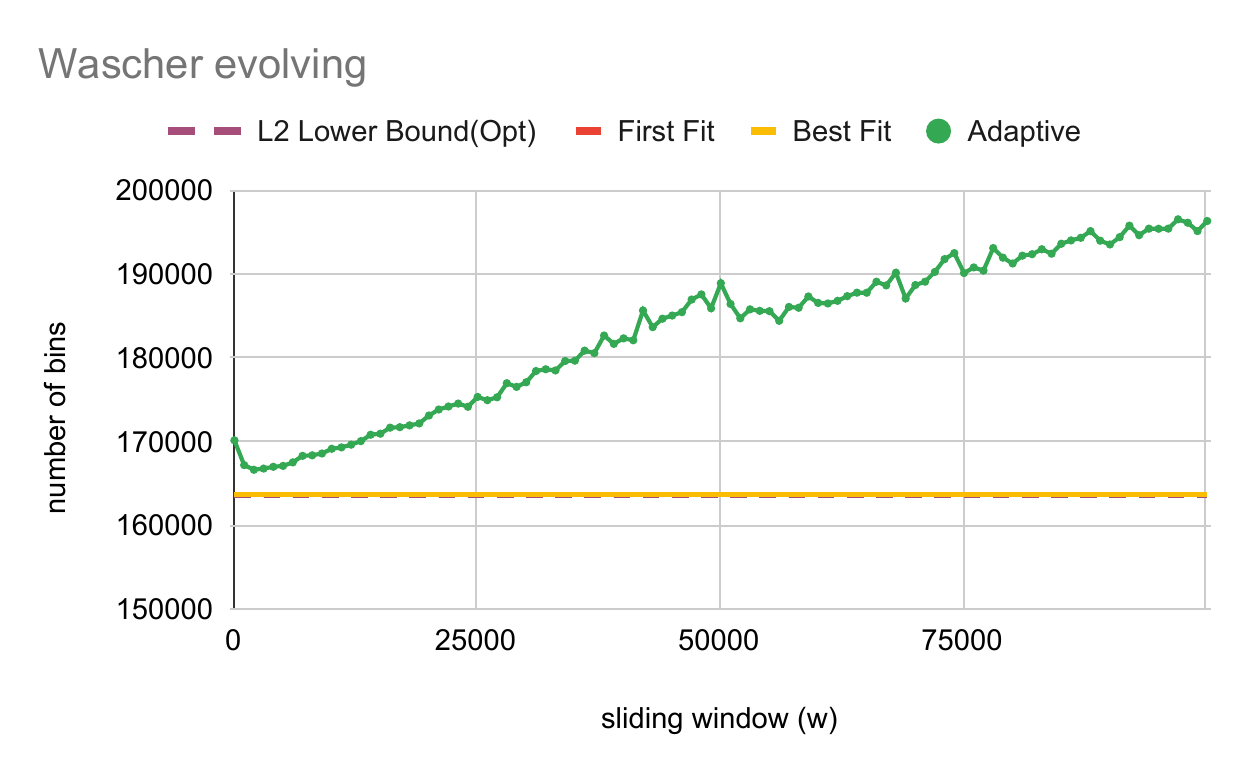}
\caption{\scriptsize W{\"a}scher benchmark from BPPLIB.}
\label{fig:dyn.Wascher}
\end{subfigure}
\caption{Number of opened bins for sequences from an evolving distribution. For the Shwerin benchmark, \firstfit and \bestfit open a similar number of bins and they practically coincide in the plot. \label{fig:evolving} }
\end{figure*}

Figure~\ref{fig:evolving} depicts the number of bins opened by \adaptive as a function of $w$ for different benchmarks. Here, we report the average cost of the algorithms over 20 randomly generated sequences.
We observe that for the Weibull and ``GI'' benchmarks, there is a relatively wide range for $w$ that
leads to performance improvement, in comparison to \firstfit and \bestfit, 
namely it suffices to choose $w \in [2100,25000]$. For ``Randomly\_Generated" and ``Schoenfield\_Hard28", the performance curve of \adaptive is similar to that on the GI benchmark, and 
\adaptive improves upon \firstfit and \bestfit when $w$ takes values in the shorter range [2000,4000].
For ``Schwerin", \adaptive always performs better, which can be explained by the discussion in Section~\ref{subsec:discussion}. For ``W{\"a}scher", \adaptive does not offer any advantage over 
\firstfit and \bestfit. However, these two baseline algorithms are remarkably close to the $L2$ lower bound, which means that they output essentially optimal packings for this benchmark, and which in turn leaves very little room for any potential improvement.


When \adaptive opens a new profile group, the predicted frequencies are updated based on the $w$ most recently packed items. These $w$ items follow a distribution that may have changed since the time a new profile group was opened. As such, the performance of \adaptive depends on the diversity of the distributions that form the benchmark. For example, for ``Schwerin", the distribution does not evolve drastically, which explains why \adaptive performs consistently better than \firstfit and \bestfit, unlike the ``W{\"a}scher" benchmark.

\section{Conclusion}
\label{sec:conclusion}

We gave the first results on the competitive analysis of online bin packing, in a setting in which the algorithm has access to learnable predictions concerning the size frequencies. 
Our approach exploits the concept of profile packing, which can be applicable in more generalized packing problems, such as two-dimensional setting studied by~\citeA{chung1982packing} and~\citeA{HuangK13} and three-dimensional setting studied by~\citeA{DBLP:conf/aaai/ZhaoS0Y021} 
and, more generally, in 
{\em vector bin packing} studied by~\citeA{azar2013tight}. These are well-studied extensions of the basic online bin packing problem, with many applications in transportation logistics and cloud computing. In these problems, a main challenge will be to leverage, or develop new offline heuristics for computing the profile packing, since the profile size increases exponentially with the dimension.

Another class of problems for which the approach may be useful is the class of {\em multicontainer} packing problems, such as multiple knapsack, bin covering, and min-cost covering. For this class of problems,~\citeA{fukunaga2007bin} gave efficient {\em bin completion} offline algorithms that can be very useful towards the design of profile-based online algorithms. Last, a further direction for future work on bin packing problems is to incorporate a {\em distributional} model of predictions, as studied by~\citeA{DBLP:conf/icml/DiakonikolasKTV21}, in which the prediction is given as the cumulative distribution function of the item size distribution. 


\section*{Acknowledgements}

This research was supported by the projects PREDICTIONS, grant ANR-23-CE48-0010, and ALGORIDAM, grant ANR-19-CE48-0016 from the French National Agency (ANR).
We acknowledge the support of the Natural Sciences and Engineering Research Council of Canada (NSERC)
[funding reference number DGECR-2018-00059].

\bibliographystyle{theapa}
\bibliography{refs-Kimia}

%
%
%
%

%
%
%

\end{document}